\documentclass{llncs}
\usepackage{epsfig}
\usepackage{algorithm}
\usepackage{amsmath}
\usepackage{amssymb}
\usepackage{bm}
\usepackage{alltt}
\newcommand{\set}[2]{\{#1\,|\,#2\}}
\newcommand{\tuple}[1]{\left<{#1}\right>}
\newcommand{\Dom}{\mbox{\sl dom\/}}

\def\bbbn{{\rm I\!N}}
\def\bbbr{{\rm I\!R}}

\begin{document}

\title{Space-Efficient Bimachine Construction Based on the Equalizer Accumulation Principle}
\author{Stefan Gerdjikov\inst{1,2} \and Stoyan Mihov\inst{2} \and Klaus U. Schulz\inst{3}}

\institute{
Faculty of Mathematics and Informatics\\
Sofia University\\
5, James Borchier Blvd., Sofia 1164, Bulgaria\\
stefangerdzhikov@fmi.uni-sofia.bg
\and
Institute of Information and Communication Technologies\\
Bulgarian Academy of Sciences\\
25A, Acad. G. Bonchev Str., Sofia 1113, Bulgaria\\
stoyan@lml.bas.bg
\and
Centrum f\"ur Informations- und Sprachverarbeitung (CIS)\\
Ludwig-Maximilians-Universit\"at M\"unchen\\
Oettingenstr. 67, 80538 M\"unchen, Germany\\
schulz@cis.uni-muenchen.de
}

\maketitle


\begin{abstract}
Algorithms for building bimachines from functional transducers found in the literature in a run of the bimachine imitate one successful path of the input transducer. Each single bimachine output exactly corresponds to the output of a single transducer transition. Here we introduce an alternative construction principle where bimachine steps take alternative parallel transducer paths into account, maximizing the possible output at each step using a joint view. The size of both the deterministic left and right automaton of the bimachine is restricted by $2^{\vert Q\vert}$ where 
$\vert Q\vert$ is the number of transducer states. Other bimachine constructions lead to larger subautomata. 
As a concrete example we present a class of real-time functional transducers with $n+2$ states for which the standard bimachine construction generates a bimachine with at least $\Theta(n!)$ states whereas the construction based on the equalizer accumulation principle leads to $2^n + n +3$ states. 
Our construction can be applied to rational functions from free monoids to ``mge monoids'', a large class of monoids including free monoids, groups, and others that is closed under Cartesian products. 
\end{abstract}

\section{Introduction}

Functional finite-state transducers and bimachines \cite{Eil74,Sakarovitch09} are devices for translating a given input sequence of symbols to a new output form. With both kinds of devices, the full class of regular string functions can be captured. However, finite-state transducers  
are more restricted in the sense that a non-deterministic behaviour is needed to realize all regular functions. In contrast, with bimachines all regular functions can be processed in a fully deterministic way. From a practical point, both models have their own advantages. For a given translation task it is often much simpler to find a non-deterministic finite-state transducer. 
Bimachines, on the other hand, are much more efficient.  General constructions that convert finite-state transducers into equivalent bimachines help to obtain both benefits at the same time.


Known algorithms for converting a functional finite-state transducer ${\cal T}$ with set of states $Q$ into an equivalent bimachine are based on the ``path reconstruction principle'': At each step of a bimachine computation, the bimachine output represents the output of a single transducer transition step. 
Furthermore, for any complete input string $w$ the sequence of bimachine outputs for $w$ is given by the sequence of outputs of ${\cal T}$ for $w$ on a specific path.  
Control of these two principles is achieved either by transforming the source transducer to be unambiguous \cite{RS97} or by using a complex notion of states for the states of the deterministic 
subautomata of the bimachine \cite{CIAA2017}. The ``enhanced'' power set constructions used to build the two subautomata have the effect that the 
size of at least one subautomaton is {\em not} bounded by $2^{\vert Q\vert}$. Recall that $2^{\vert Q\vert}$ is the 
bound obtained for a standard power set determinization. 

Here we introduce a new construction that only needs a conventional power set construction for the states of {\em both} subautomata of the bimachine. As a consequence, the size of both subautomata is bound by $2^{\vert Q\vert}$. In the new construction, the output of a single bimachine step takes into account the outputs of several parallel alternative transducer transitions, in a way to be explained. Using a principle called ``equalizer accumulation'' the joint view on all relevant transducer transitions leads to a kind of maximal output for the bimachine. At the same time this ``joint look'' at parallel paths of the transducer guarantees that the complete bimachine output for an input string $w$ is identical to the combined output of the transducer on {\em any} path for $w$. 

The new construction is not restricted to strings as bimachine output. 
As a second contribution of the paper we introduce and study the class of ``monoids with most general equalizers'' (mge monoids).
We show that this class includes all free monoids, groups, the tropical semiring, and others and is closed under Cartesian products. 
The aforementioned principle of ``equalizer accumulation'' is possible for all mge monoids. 

The paper has the following structure. We start with formal preliminaries in Section~\ref{SecFormalPreliminaries}. 
In Section~\ref{SecMGE} we introduce mge monoids and study formal properties of this class. We define the principle of 
equalizer accumulation and show that equalizer accumulation is possible for all mge monoids. 
In Section~\ref{SecFunctionalityTest} we give an algorithm for deciding the functionality of a transducer with outputs in a mge monoid. 
In Section~\ref{SecBimachineConstruction} we introduce the new algorithm for converting a functional finite-state transducer 
with outputs in a mge monoid into an equivalent bimachine. 
In Section~\ref{SectionComparison} we present a class of real-time functional transducers with $n+2$ states for which the standard bimachine construction generates a bimachine with at least $\Theta(n!)$ states whereas the construction based on the equalizer accumulation principle leads to $2^n + n +3$ states only.
We finish with a short conclusion in Section~\ref{SecConclusion}.

\section{Formal Preliminaries}\label{SecFormalPreliminaries}

\begin{definition}\label{DefMonoid}
A {\em monoid} ${\cal M}$   is a  triple $\tuple{M,\circ,e}$, where $M$ is a non-empty set, $\circ :M \times M \rightarrow M$ is a total binary function that is associative, i.e.  $\forall a,b,c \in M: a \circ (b \circ c) = (a \circ b) \circ c$  and  $e \in M$ is the unit element, i.e. $\forall a \in M: a \circ e = e \circ a = a$.  Products $a \circ b$ are also written $a b$.
\end{definition}
\begin{definition}\label{DefCancellation}
The monoid ${\cal M}$ supports {\em left cancellation} if $\forall a,b,c \in M : c \circ a = c \circ b \rightarrow a=b$.
The monoid ${\cal M}$ supports {\em right cancellation} if $\forall a,b,c \in M : a \circ c = b \circ c \rightarrow a=b$.
\end{definition}
We list some wellknown notions used in the paper.
An {\em alphabet} $\Sigma$ is a finite set of symbols. 
Words of length $n\geq 0$ over an alphabet $\Sigma$ are introduced as usual and written $\tuple{a_1,\ldots,a_n}$ or simply $a_1\ldots a_n$ ($a_i \in \Sigma$). The {\em concatenation} of two words $u=\tuple{a_1,\ldots,a_n}$ and $v=\tuple{b_1,\ldots,b_m}\in \Sigma^\ast$ is $u \cdot v = \tuple{a_1,\ldots,a_n,b_1,\ldots,b_m}.$  The unique word of length $0$ (``empty word'') is written $\varepsilon$. The set of prefixes of a word $w$ is introduced as usual. $\Sigma^\ast$ denotes the set of all words over $\Sigma$. The set  $\Sigma^\ast$ with concatenation as monoid operation and the empty  word $\varepsilon$  as  unit element is called the {\em free monoid} over $\Sigma$. The expression $u^{-1}t$ denotes the word $v$ if $u$ is a prefix of  $t$ and $t=u\cdot v$, otherwise $u^{-1}t$ is undefined. 

\begin{proposition}
Let ${\cal M} = \tuple{M,\circ_M,e_M}$ and ${\cal N} = \tuple{N,\circ_N,e_N}$ be two monoids. Let ${\cal M}\times {\cal N} := \tuple{M\times N,\circ_{M\times N},\tuple{e_M,e_N}}$, where $\circ_{M\times N} : (M\times N)^2 \rightarrow M\times N$  is the total function such that $\tuple{a_1,a_2}  \circ_{M\times N} \tuple{b_1,b_2} := \tuple{a_1 \circ_M b_1, a_2 \circ_N b_2}$. Then ${\cal M}\times {\cal N}$ is a monoid.
\end{proposition}
The monoid ${\cal M}\times {\cal N}$ is called the {\em Cartesian product} of the monoids ${\cal M}$ and ${\cal N}$.


\begin{definition}\label{DefFSA}
A {\em monoidal finite-state automaton} is a tuple of the form
${\cal A} = \tuple{ {\cal M},Q,I,F,\Delta}$ where 
\begin{itemize}
\item
  ${\cal M}=\tuple{M,\circ,e}$ is a monoid, 
\item
  $Q$ is a finite set called the set of states, 
\item
  $I\subseteq Q$ is the set of initial states, 
\item 
  $F\subseteq Q$ is the set of final states, and 
\item
  $\Delta \subseteq Q\times M \times Q$ 
  is a finite set of transitions called the transition relation. 
\end{itemize}
\end{definition}
\begin{definition}
Let ${\cal A} = \tuple{ {\cal M},Q,I,F,\Delta}$ be a monoidal finite-state automaton. 
The {\em generalized transition relation} $\Delta^\ast$ is defined as the
smallest subset of $Q \times M \times Q$ with the following closure 
properties:
\begin{itemize}
\item
  for all $q\in Q$ we have $\tuple{q,e,q} \in \Delta^\ast$.
\item
  For all $q_1, q_2, q_3\in Q$ and $w, a \in M$:
  if $\tuple{q_1,w,q_2} \in \Delta^\ast$ and $\tuple{q_2,a,q_3} \in \Delta$,  
  then also $\tuple{q_1,w \circ a,q_3} \in \Delta^\ast$.
\end{itemize}
The monoidal language  
{\em accepted} (or {\em recognized}) by ${\cal A}$ is 
$$L({\cal A}) :=\set{w\in M}{\exists p \in I\ \exists q \in F: \tuple{p,w,q}\in \Delta^\ast}.$$
\end{definition}
\begin{definition}
Let ${\cal A} = \tuple{ {\cal M},Q,I,F,\Delta}$ be a monoidal finite-state automaton. 
A {\em proper path} of ${\cal A}$ is a finite sequence of $k>0$ transitions
$$
\pi = \tuple{q_0,a_1,q_1}\tuple{q_1,a_2,q_2}\ldots\tuple{q_{k-1},a_k,q_k}
$$
where $\tuple{q_{i-1},a_i,q_i} \in \Delta$ for $i=1\ldots k$.
The number $k$ is called the {\em length} of $\pi$, we say that $\pi$ starts in $q_0$ and ends in $q_k$. 
States $q_0,\ldots,q_k$ are {\em the states on the path $\pi$}. If 
$q_0=q_k$, then the path is called a {\em cycle}. 
The monoid element $w=a_1 \circ \ldots \circ a_k$ is called the {\em label} of $\pi$. 
We denote $\pi$ as
$\pi = q_0 \rightarrow^{a_1} q_1 \ldots \rightarrow^{a_k} q_k$.
A {\em successful path} is a path starting in an initial state and ending in a final state. 
\end{definition}
A state $q\in Q$  is {\em accessible} if $q$ is the ending of a path of ${\cal A}$ starting from an initial state. A state $q\in Q$  is {\em co-accessible} if $q$ is the starting of a path of ${\cal A}$ ending in a final state. 
A monoidal finite-state automaton ${\cal A}$ is {\em trimmed} iff each state $q\in Q$ is accessible and co-accessible. 
A monoidal finite-state automaton ${\cal A}$ is {\em unambiguous} iff for every element $m \in M$ there exists at most one successful path in ${\cal A}$ with label $m$. 
\begin{definition}\label{DefExtExt}
Let ${\cal A} = \tuple{ {\cal M},Q,I,F,\Delta}$ be a monoidal finite-state automaton. Then the {\em $e$-extended automaton} for ${\cal A}$ is the monoidal finite-state automaton ${\cal A}^{ext} = \tuple{ {\cal M},Q,I,F,\Delta^{ext}}$, where $\Delta^{ext} = \Delta \cup \set{\tuple{q,e,q}}{q \in Q}$.
\end{definition}
Clearly the $e$-extended automaton for ${\cal A}$ is equivalent to ${\cal A}$ i.e. $L({\cal A}^{ext})=L({\cal A})$.
\begin{definition}
A monoidal finite-state automaton ${\cal A} = \tuple{ {\cal M},Q,I,F,\Delta}$  is a {\em mo\-noidal finite-state transducer} iff ${\cal M}$  can be represented as the Cartesian product of a free monoid $\Sigma^\ast$ with another monoid ${\cal M}'$,  i.e ${\cal M} = \Sigma^\ast \times {\cal M}'$ and $\Delta\subseteq Q\times (\Sigma \cup \{\varepsilon\}\times M') \times Q$.
A monoidal finite-state transducer ${\cal A} = \tuple{ \Sigma^\ast \times {\cal M}',Q,I,F,\Delta}$ is  {\em real-time} if $\Delta\subseteq Q\times (\Sigma\times M') \times Q$. A monoidal finite-state transducer ${\cal A}$ is  {\em functional} if $L({\cal A})$ is a function. In this case we denote the function recognized by ${\cal A}$ as $O_{\cal A} : \Sigma^\ast \rightarrow M'$.
\end{definition}
Let ${\cal M}$ be a monoid. A language $L \subseteq M$ is {\em rational} iff it is accepted by a monoidal finite-state automaton. A {\em rational function} is a rational language that is a function.
\begin{definition}\label{DefBimachine}
A {\em bimachine} is a tuple ${\cal B}=\tuple{{\cal M},{\cal A}_L,{\cal A}_R,\psi}$, where:
\begin{itemize}
\item ${\cal A}_L = \tuple{\Sigma,L,s_L,L,\delta_L}$ and ${\cal A}_R = \tuple{\Sigma,R,s_R,R,\delta_R}$ are deterministic finite-state automata called the {\em left} and {\em right automaton} of the bimachine;
\item ${\cal M} = \tuple{M,\circ,e}$ is the output monoid and $\psi : (L \times \Sigma \times R) \rightarrow M$ is a partial function called the {\em  output function}.
\end{itemize}
Note that all states of ${\cal A}_L$ and ${\cal A}_L$ are final. The function $\psi$ is naturally extended to the  {\em generalized output function}\index{Generalized output function of a bimachine} $\psi^\ast$ as follows:
\begin{itemize}
\item $\psi^\ast(l,\varepsilon,r)=e$ for all $l\in L, r\in R$;
\item $\psi^\ast(l,t\sigma,r)= \psi^\ast(l,t,\delta_R(r,\sigma)) \circ  \psi(\delta^\ast_L(l,t),\sigma,r)$ for $l\in L, r\in R, t\in\Sigma^\ast, \sigma\in\Sigma$.
\end{itemize}
The {\em function represented by the bimachine} is 
$$O_{\cal B}:\Sigma^\ast \rightarrow M: t  \mapsto \psi^\ast(s_L,t,s_R).$$
If $O_{\cal B}(t)=t'$ we say that the bimachine ${\cal B}$ {\em  translates $t$ into $t'$}. 
\end{definition}

\section{Monoids with most general equalizers}\label{SecMGE}

Monoids with most general equalizers (mge monoids), to be introduced below, generalize both free monoids and groups. 
In this section we study properties that provide the basis for the bimachine construction presented afterwards.  

\begin{definition}\label{Defmostgeneralequalizers}
Let ${\cal M}=\tuple{M,\circ,e}$ be a monoid, let $n\geq 1$. 
A tuple $\tuple{ m_1,\ldots m_n}$ in $M^n$ is called 
{\em equalizable} iff there exists $\tuple{x_1,\ldots,x_n} \in M^n$ such that 
$$m_1 x_1=\ldots = m_n x_n.$$ 
In this situation 
$\tuple{ x_1,\ldots,x_n}$ is called an {\em equalizer} for $\tuple{ m_1,\ldots m_n}$. 
If 
$\tuple{ x_1,\ldots,x_n}$ is an equalizer for $\tuple{ m_1,\ldots m_n}$ and $x\in M$, then  
$\tuple{ x_1 x,\ldots,x_n x}$ is called an {\em instance} of $\tuple{ x_1,\ldots,x_n}$. An equalizer
$\tuple{ x_1,\ldots,x_n}$ for $\tuple{ m_1,\ldots m_n}$ is called a {\em most general equalizer} (mge) for 
$\tuple{ m_1,\ldots m_n}$ if each equalizer
$\tuple{ x_1',\ldots,x_n'}$ for $\tuple{ m_1,\ldots m_n}$ is an instance of $\tuple{ x_1,\ldots,x_n}$. By $Eq^{(n)}(m_1,\ldots,m_n) \subseteq M^n$ we denote the set of all equalizers of $\tuple{ m_1,\ldots m_n}$.
\end{definition}
As a matter of fact for $n>1$ in general not all tuples $\tuple{ m_1,\ldots,m_n}$ in a monoid are equalizable. There are monoids with equalizable tuples that do not have any mge. 
In our context, most general equalizers become relevant when considering 
intermediate outputs $\tuple{ m_1,\ldots,m_n}$ of a functional transducer obtained for the same input word 
$w$ on distinct initial paths. If for some continuation $v$ of $w$ each path can be continued to a final state, on a path with input $v$, then the completed outputs 
$m_1m_1',\ldots,m_nm_n'$ must be equal. If $\tuple{ x_1,\ldots,x_n}$ is a most general equalizer of $\tuple{ m_1,\ldots,m_n}$, then we 
can safely output $m_1x_1=\ldots = m_nx_n$, ``anticipating'' {\em necessary} output. Below we use this idea for bimachine construction. 
%
%
We introduce a special class of monoids where 
equalizable pairs always have mge's. 
\begin{definition}\label{DefceMonoid}
The monoid ${\cal M}=\tuple{M,\circ,e}$ is a {\em mge monoid} iff ${\cal M}$ has right cancellation and 
each equalizable pair $\tuple{ m_1,m_2} \in M^2$ has a mge.
\end{definition}
\begin{example}\label{Ex1mgeFree}
(a) Let $\Sigma$ be an alphabet. The free monoid $\tuple{\Sigma^\ast,\cdot,\varepsilon}$ for alphabet $\Sigma$ is 
a mge monoid. In fact, a pair of words $\tuple{ u,v}$ is equalizable iff $u$ is a prefix of $v=u x$ or $v$ is a prefix of 
$u=vx$. In the former (latter) case $\tuple{ x,\varepsilon}$ (resp. $\tuple{ \varepsilon,x}$) is a mge for $\tuple{ u,v}$. \\
(b) Another example of a mge monoid is the additive monoid of nonnegative real numbers, $(\bbbr_{\geq 0},+,0)$. Consider a pair of non-negative reals $\tuple{m,n}$.
Let $M=\max(m,n)$, thus $M-m,M-n\in \mathbb{R}_{\ge 0}$.First, $m+(M-m)=n+(M-n)=M$ shows that $\tuple{M-m,M-n}$ is an equalizer for $\tuple{m,n}$. Next, if $\tuple{x,y}$ is an arbitrary equalizer for $\tuple{m,n}$, then $m+x=n+y=N$. Since $x$ and $ y$ are nonnegative $N\ge M$ and $N-M\in \mathbb{R}_{\ge 0}$. Thus, $x=(M-m) +(N-M)$ and $y=(M-n)+(N-M)$. This proves that
$\tuple{M-m,M-n}$ is a mge of $\tuple{m,n}$.\\
(c) Both Example (a) and Example (b) are special cases of sequentiable structures in the sense of 
Gerdjikov \& Mihov \cite{LATA2017}. It can be shown that all sequentiable structures are mge monoids. This is a simple consequence of Proposition~2 in \cite{LATA2017}. 
\end{example}
%
%
\begin{example}\label{Ex2mgeGroup}
Let ${\cal G}=\tuple{G,\cdot,1,^{-1}}$ be a group. Then $\tuple{G,\cdot,1}$ is a mge monoid. For each pair 
$\tuple{g,h}$ of group elements, both $\tuple{1,h^{-1}\cdot g}$ and 
$\tuple{ g^{-1}\cdot h,1}$ are mge's. 
\end{example}
\begin{lemma}\label{LemmaCartesianMGE}
Let ${\cal M}=\tuple{M,\circ_M,e_m}$ and ${\cal N}=\tuple{N,\circ_N,e_n}$ be mge monoids. Then the Cartesian product  
${\cal M}\times {\cal N}$ is a mge monoid. 
\end{lemma}
\begin{proof}
Let $\tuple{\tuple{m_1,n_1},\tuple{ m_2,n_2}}$ be an equalizable pair, let 
$\tuple{\tuple{ u_1,v_1},\tuple{ u_2,v_2}}$ be an equalizer. Then 
$$\tuple{ m_1u_1,n_1v_1}=\tuple{ m_1,n_1}\tuple{ u_1,v_1}=\tuple{ m_2,n_2}\tuple{ u_2,v_2}=\tuple{ m_2u_2,n_2v_2},$$
which implies that 
\begin{eqnarray*}
m_1u_1 &=& m_2u_2\\
n_1v_1 &=& n_2v_2.
\end{eqnarray*}
Both $\tuple{m_1,m_2}$ and $\tuple{n_1,n_2}$ are equalizable. Let $\tuple{x_1,x_2}$ and $\tuple{y_1,y_2}$ respectively denote
mge's. For some $c\in M$ and $d\in N$ we have 
$u_i=x_ic$ and $v_i=y_id$ ($i=1,2$). The equations
$$\tuple{ m_1,n_1}\tuple{ x_1,y_1} = \tuple{ m_1x_1,n_1y_1}=\tuple{ m_2x_2,n_2y_2}=\tuple{ m_2,n_2}\tuple{ x_2,y_2}.$$
show that $\tuple{\tuple{ x_1,y_1},\tuple{ x_2,y_2}}$ is an equalizer for 
$\tuple{\tuple{ m_1,n_1},\tuple{ m_2,n_2}}$. Furthermore, 
$$\tuple{\tuple{ x_1,y_1}\tuple{ c,d},\tuple{ x_2,y_2})\tuple{ c,d}}=\tuple{\tuple{ x_1c,y_1d},\tuple{ x_2c,y_2d}}=\tuple{\tuple{ u_1,v_1},\tuple{ u_2,v_2}}.$$
Hence each equalizer is an instance of $\tuple{\tuple{ x_1,y_1},\tuple{ x_2,y_2}}$ and 
$\tuple{\tuple{ x_1,y_1},\tuple{ x_2,y_2}}$ is a mge.\qed
\end{proof}
\begin{lemma}\label{LemmaLeftCancelation}
Let ${\cal M}=\tuple{M,\circ,e}$ be a mge monoid. Then ${\cal M}$ has left cancellation and for each element $m \in M$ the pair $\tuple{e,e}$ is a mge of $\tuple{m,m}$.
\end{lemma}
\begin{proof}
Let $a,b,c$ be arbitrary elements of $M$ such that $ca=cb$. Clearly the pair $\tuple{e,e}$ is an equalizer of $\tuple{c,c}$ and therefore there exists a mge $\tuple{x,y}$ of $\tuple{c,c}$. Hence there is some $d \in M$ such that $x d = e = y d$. Therefore, by the right cancellation property we get that $x = y$. On the other hand $\tuple{a,b}$ is also an equalizer of $\tuple{c,c}$ and therefore there is some $d' \in M$ such that $x d' = a$ and $y d' = b$. Since $x=y$ we obtain that $a=b$. 

\noindent Let $\tuple{a,b}$ be an equalizer of $\tuple{m,m}$. Then $ma=mb$ and thus $a=b$. Hence $\tuple{e,e}$ is a mge of $\tuple{m,m}$.\qed
\end{proof}
\begin{definition}\label{DefJointequalizers}
Let ${\cal M}$ be a monoid, let $n\geq 1$. A tuple $\tuple{ x_1,\ldots,x_n}$ 
is called a {\em joint equalizer} for $\tuple{ m_1,\ldots m_n}$ and 
$\tuple{ m'_1,\ldots m'_n}$ if $\tuple{ x_1,\ldots,x_n}$ is an equalizer 
both for $\tuple{ m_1,\ldots m_n}$ and for $\tuple{ m_1,\ldots m_n}$.
\end{definition}
Note that in the situation of the definition we have $m_1x_1=\ldots = m_nx_n$ and 
$m'_1x_1=\ldots =m'_nx_n$ but we do {\em not} demand that $m_ix_i=m_i'x_i$ ($1\leq i\leq n$). 
\begin{lemma}\label{LemmaJointEq}
Let ${\cal M}=\tuple{M,\circ,e}$ be a mge monoid.
If $\tuple{ m_1,\ldots,m_l}$ and $\tuple{ n_1,\ldots,n_l}$ have a joint equalizer, then 
each mge for $\tuple{ m_1,\ldots,m_l}$ is a mge for $\tuple{ n_1,\ldots,n_l}$ and vice versa. Furthermore $Eq^{(l)}(m_1,\ldots,m_l) = Eq^{(l)}(n_1,\ldots,n_l)$.
\end{lemma}
\begin{proof}
Let $\tuple{k_1,\ldots,k_l}$ denote a joint equalizer for $\tuple{ m_1,\ldots,m_l}$ and $\tuple{ n_1,\ldots,n_l}$. 
Let $(x_1,\ldots,x_l)$ be a mge for $\tuple{ m_1,\ldots,m_l}$, let $(y_1,\ldots,y_l)$ be a mge for $\tuple{ n_1,\ldots,n_l}$. 
Then there exists some $c\in M$ such that $\tuple{k_1,\ldots,k_l}=(x_1c,\ldots,x_lc)$. 
We obtain $$n_1x_1c= \ldots =n_lx_lc.$$
Using right cancellation we get $n_1x_1=\ldots=n_lx_l$, which shows that $\tuple{x_1,\ldots,x_l}$ is an equalizer for $\tuple{ n_1,\ldots,n_l}$ and an instance of $\tuple{y_1,\ldots,y_l}$. 
Symmetrically we see that $\tuple{y_1,\ldots,y_l}$ is an instance of $\tuple{x_1,\ldots,x_l}$. It follows that both $\tuple{x_1,\ldots,x_l}$ and $\tuple{y_1,\ldots,y_l}$ are mge's for 
$\tuple{ m_1,\ldots,m_l}$ and for $\tuple{ n_1,\ldots,n_l}$.\qed
\end{proof}
\begin{definition}
Let ${\cal M}=\tuple{M,\circ,e}$ be a mge monoid. An element $m\in M$ is {\em invertible} if there exists an element $n\in M$ such that $mn=e$. It follows from the left cancellation property that $n$ is unique, we write 
$m^{-1}$ for $n$ and call $m^{-1}$ the {\em inverse} of $m$. 
\end{definition}
If $m$ has an inverse $m^{-1}$, then 
$mm^{-1}m=m=me$ and left cancellation shows that $m^{-1}m=e$. Hence ``right'' inverses are left inverses and vice versa. 

\begin{lemma}\label{lemmaSolvability}
Let ${\cal M}=\tuple{M,\circ,e}$ be a mge monoid. Let $m,n\in M$ be given and assume that the equation
$mx=n$ has a solution. Then the solution is unique and if $\tuple{x_1,x_2}$ is a mge of $\tuple{m,n}$ then $x=x_1 x_2^{-1}$. 
\end{lemma}
\begin{proof} Uniqueness of solutions follows directly from left cancellation. 
If $mx=n$ has a solution $x$, then $\tuple{x,e}$ is an equalizer for $\tuple{m,n}$. 
 Then there exists an element $d\in M$ such that 
$x_1d=x$ and $x_2d=e$. Therefore $d$ is the inverse of $x_2$ i.e. $d :=x_2^{-1}$. Now $mx_1x_2^{-1}=nx_2x_2^{-1}=n$. The product 
$x_1x_2^{-1}$ represents the unique solution for $mx=n$.\qed
\end{proof}
\begin{definition}\label{DefComputational}  
A mge monoid ${\cal M}=(M,\circ,e)$ is {\em effective} if 
(i) M can be represented as a recursive subset of $\bbbn$ and the monoid operation is computable (ii) the equality of two monoid elements is decidable, (iii) it is decidable whether a 
given pair $\tuple{ m,m'}\in M^2$ is equalizable, there is a computable function $\eta \in M^2 \rightarrow M^2$, such that for each equalizable pair $\tuple{ m,m'}$ always $\eta(m,m')$ is a mge of $\tuple{ m,m'}$, and
(iv) for every invertible element $m\in M$ the inverse $m^{-1}$ is computable.
\end{definition}
It is simple to see that the class of effective mge monoids is closed under Cartesian products and that
the above examples of mge monoids are effective under natural assumptions\footnote{The tropical semiring is not effective, but the finite set of transition outputs of a given transducer can be embedded in an effective mge submonoid of the tropical semiring.}.
\begin{corollary}\label{CorSolvability}
Let ${\cal M}=\tuple{M,\circ,e}$ be an effective mge monoid. Let $m,n\in M$ be given and assume that the equation
$mx=n$ has a solution. Then we may effectively compute the solution.
\end{corollary}
\begin{lemma}\label{LemmaCmgesinMGEmonoids}
Let ${\cal M}=\tuple{M,\circ,e}$ be a mge monoid. Then 
for each $n\geq 1$, every equalizable tuple $\tuple{m_1,\ldots,m_n}$ has a mge. 
\end{lemma}
\begin{proof}
We use induction on $n$. For $n=1$ note that $e$ is always a mge for $\tuple{ m}$. For $n=2$ equalizable pairs have mge's 
by definition of mge monoids. 
Let $n\geq 2$ and 
$\tuple{ m_1,\ldots,m_n,m_{n+1}}$
be equalizable, say 
$m_1v_1=\ldots = m_nv_n=m_{n+1}v_{n+1}$. 
By induction hypothesis there exists a mge $\tuple{ x_1,\ldots,x_n}$ for $\tuple{ m_1,\ldots,m_n}$. 
For some $c\in M$ we have $x_ic=v_i$ for $1\leq i\leq n$. 

Since $m_n$ and $m_{n+1}$ are equalizable there exists a mge $\tuple{y_n,y_{n+1}}$ and some $d\in M$ such that $y_n d=v_n$, $y_{n+1} d=v_{n+1}$. 
Now we have two representations for $v_n$:
$$x_n c=v_n=y_n d.$$
Hence $\tuple{x_n,y_n}$ is equalizable. Let $\tuple{z_x,z_y}$ be a mge. Since $\tuple{c,d}$ is an instance there exists an $f\in M$ such that $c=z_x f$, $d=z_y f$.
Now we can represent the full equalizer $\tuple{v_1,\ldots,v_n,v_{n+1}}$ in the following form
\begin{eqnarray*}
v_1 &=& x_1c=x_1z_xf\\
&\ldots& \\
v_n &=& x_nc=x_nz_xf\\
v_{n+1} &=& y_{n+1}d=y_{n+1}z_yf
\end{eqnarray*}
as an instance of $\tuple{x_1z_x,\ldots,x_n z_x,y_{n+1} z_y}$. As a matter of fact the latter is an equalizer for $\tuple{ m_1,\ldots,m_n,m_{n+1}}$. 
Since $\tuple{v_1,\ldots,v_n,v_{n+1}}$ was {\em any} equalizer for $\tuple{ m_1,\ldots,m_n,m_{n+1}}$
it follows that $\tuple{x_1z_x,\ldots,x_n z_x,y_{n+1} z_y}$ is a mge for $\tuple{ m_1,\ldots,m_n,m_{n+1}}$.\qed
%
\end{proof}
%
%
The proof of Lemma~\ref{LemmaCmgesinMGEmonoids} provides us with an effective way to compute mge's. 
The following corollaries describe the principle of equalizer accumulation mentioned in the introduction. 
\begin{corollary}\label{CorollaryCombinationPairwisecmges1}
Let ${\cal M}$ be an effective mge monoid. If $\tuple{m_1,\ldots,m_n} \in M^n$ are equalizable, then ${\mu}^{(n)}(m_1,\ldots,m_n)$ is the mge of 
$m_1,\ldots,m_n$, where
${\mu}^{(n)}:M^n\rightarrow M^n$ is defined as:
\begin{enumerate}
\item $\mu^{(1)}(m)=e$ for every $m\in M$,
\item $\mu^{(2)}(m_1,m_2)=\eta(m_1,m_2)$ for each equalizable pair $\langle m_1,m_2\rangle$,
\item If $\mu^{(n)}$ is given, then for every equalizable tuple $\langle m_1,\dots,m_{n+1}\rangle \in M^{n+1}$ the $i$-th coordinate of $\mu^{(n+1)}(m_1,\dots,m_{n+1})$, denoted $\mu^{(n+1)}_i(m_1,\dots,m_{n+1})$,  is given by
\begin{eqnarray*}
&&\mu_i^{(n+1)}(m_1,\dots,m_{n+1})=\\
&&\begin{cases}
\mu_i^{(n)}(m_1,\ldots,m_n) \circ z_x\text{ for } i\le n\\
\eta_2(m_n,m_{n+1}) \circ z_y \text{ for } i=n+1
\end{cases}
\end{eqnarray*}
where $\langle z_x,z_y\rangle := \eta_(\mu_n^{(n)}(m_1,\ldots,m_n),\eta_1(m_n,m_{n+1}))$.
\end{enumerate} 
\end{corollary} 
The next corollary shows that we can express the function $\mu$ using only pairwise mge's.
\begin{corollary}\label{CorollaryCombinationPairwisecmges2} In the settings of Corollary~\ref{CorollaryCombinationPairwisecmges1}
$$\mu^{(n)}(m_1,\ldots,m_n) = \gamma^{(n)}(\eta(m_1,m_2),\eta(m_2,m_3),\ldots,\eta(m_{n-1},m_n)),$$
where $\gamma^{(n)} : (M^2)^{n-1} \rightarrow M^n$ is a partial function defined inductively:
\begin{enumerate}
\item $\gamma^{(1)}=e$,
\item $\gamma^{(2)}(X)=X$ for every $X \in M^2$,
\item \begin{eqnarray*}
&&\gamma_i^{(n+1)}(X^{(1)},\dots,X^{(n)})=\\
&&\begin{cases}
\gamma_i^{(n)}(X^{(1)},\dots,X^{(n-1)}) \circ \eta_1(\gamma_n^{(n)}(X^{(1)},\dots,X^{(n-1)}),X_1^{(n)})\text{ for } i\le n\\
X_2^{(n)} \circ \eta_2(\gamma_n^{(n)}(X^{(1)},\dots,X^{(n-1)}),X_1^{(n)}) \text{ for } i=n+1
\end{cases}
\end{eqnarray*}
\end{enumerate} 
\end{corollary}
In particular, if two tuples $\tuple{m_1',\dots,m_n'}$ and $\tuple{m_1'',m_2'',\dots,m_n''}$ have the property:
\begin{equation*}
\eta(m_i',m_{n}')\text{ is a mge for } \tuple{m_i'',m_{n}''} \text{ for each } i<n,
\end{equation*}
then $\mu^{(n)}(m_1',m_2',\dots,m_n')=\mu^{(n)}(m_1'',m_2'',\dots,m_n'')$.
\begin{remark}\label{pure_unity_case}
Note that if $\eta(e,e)=\tuple{e,e}$, then 
$$\gamma^{(n)}(\tuple{e,e},\tuple{e,e},\dots,\tuple{e,e})=\tuple{e,e,\dots,e}$$
\end{remark}

\section{Squared automata and functionality test for mge monoids}\label{SecFunctionalityTest}
In this section we characterize the class of functional transducers with outputs in an effective mge monoid and give an efficient algorithm for testing the functionality. We make use of the squaring transducer approach presented in \cite{BCPS03}.

Let ${\cal M}=(M,\circ,e)$ be an effective mge monoid and $\eta$ be a computable function computing the mge. 
We start with some simple observations that enables us to restrict attention to real-time transducers.
\begin{lemma}\label{epsilon-cycles}
Let ${\cal T}=\tuple{\Sigma^*\times {\cal M},Q,I,F,\Delta}$ be a trimmed functional transducer, let $p\in M$. If
$\tuple{p,\tuple{\varepsilon,m},p}\in \Delta^*$,
then $m=e$.
\end{lemma}
\begin{proof}
Since ${\cal T}$ is trimmed state $p$ is both accessible and co-accessible. There exist $u_1,u_2\in \Sigma^*$, $i\in I$, $f\in F$, and $m_1,m_2\in M$ such that
\begin{equation*}
\tuple{i,\tuple{u_1,m_1},p}\in \Delta^* \text{ and } \tuple{p,\tuple{u_2,m_2},f}\in \Delta^*.
\end{equation*}
Therefore $\tuple{u_1u_2,m_1m_2}\in {\cal L}({\cal T})$. Using the cycle $\tuple{p,\tuple{\varepsilon,m},p}\in \Delta^*$, we see that $\tuple{i,\tuple{u_1u_2,m_1mm_2},f}\in \delta^*$. Hence $\tuple{u_1u_2,m_1mm_2}\in {\cal L}({\cal T})$. 

Finally, since ${\cal T}$ is functional we have $m_1m_2=m_1mm_2$. Since ${\cal M}$ satisfies the left and the right cancellation properties we obtain $e=m$.\qed
\end{proof}

\begin{remark}\label{real-time1} 
A (trimmed) transducer ${\cal T}=\tuple{\Sigma^*\times {\cal M},Q,I,F,\Delta}$ contains 
a cycle with label $(\varepsilon,m)$ such that $m\neq e$ if and only if ${\cal T}$ contains a 
cycle $\pi$ of length less or equal to $|Q|$, with label $(\varepsilon,m')$ such that $m'\neq e$. 
Thus we can efficiently check whether there exists a cycle
\begin{equation*}
\tuple{p,\tuple{\varepsilon,m},p}\in \Delta^* \text{ with } m\neq e.
\end{equation*}
If this is the case, we reject that ${\cal T}$ is functional. 
\end{remark}
\begin{remark}\label{real-time2}
If every $\varepsilon$-cycle in ${\cal T}$ is an $(\varepsilon,e)$-cycle, then we can apply a specialised $\varepsilon$-closure procedure (see e.g. \cite{RS97}) and convert ${\cal T}$ into a real-time transducer, ${\cal T}'$, with
\begin{equation*}
L({\cal T}) \cap (\Sigma^+ \times {\cal M})=L({\cal T}') \cap (\Sigma^+ \times {\cal M}).
\end{equation*}
In addition we can compute the finite set:
\begin{equation*}
L_{\varepsilon}({\cal T})= \{m\,|\, \exists i\in I,f\in F (\tuple{i,\tuple{\varepsilon,m},f}\in \Delta^*\}.
\end{equation*}
If $\vert L_{\varepsilon}({\cal T})\vert > 1$, then we reject that ${\cal T}$ is functional.
Otherwise, the functionality test for ${\cal T}$ boils down to check whether the real-time transducer ${\cal T}'$ is functional.
\end{remark}

Due to Remarks~\ref{real-time1} and~\ref{real-time2} for the functionality test we may assume 
that ${\cal T}$ is a real-time transducer. 
(At the end of this section we add a note how the test can be efficiently applied to an arbitrary trimmed transducer directly without such a preprocessing step.) 
As a preparation we need the following structure. 
\begin{definition}
Let ${\cal T}=\tuple{\Sigma^\ast\times{\cal M},Q,I,\Delta,F}$ be a  monoidal finite-state transducer, let 
$\Delta_{2}\subseteq (Q\times Q)\times ({\cal M}\times {\cal M}) \times (Q\times Q)$ be finite. 
Then $$\tuple{{\cal M}\times {\cal M},Q\times Q,I\times I,F\times F,\Delta_2}$$ is a {\em squared output automaton} for ${\cal T}$ if the following holds:
\begin{eqnarray*}
\tuple{\tuple{p_1,q_1},\tuple{m,n},\tuple{p_2,q_2}}\in \Delta_2^* \iff \exists u\in \Sigma^*&&(\tuple{p_1,\tuple{u,m},q_1}\in \Delta^* \\&&\text{ and } \tuple{p_2,\tuple{u,n},q_2}\in \Delta^*).
\end{eqnarray*}
\end{definition}

\begin{proposition}
Let ${\cal T}=\tuple{\Sigma\times{\cal M},Q,I,\Delta,F}$ be a  real-time transducer. 
Then the monoidal finite-state-automaton ${\cal A}^{sq}$ defined as:
\begin{eqnarray*}
{\cal A}^{sq} &=& \tuple{{\cal M}\times {\cal M},Q\times Q,I\times I,F\times F,\Delta_2}\\
\Delta_{sq} & = & \{\tuple{\tuple{p_1,p_2},\tuple{m_1,m_2},\tuple{q_1,q_2}}\,|\, \exists a\in \Sigma:\tuple{p_i,\tuple{a,m_i},q_i}\in \Delta \text{ for } i=1,2\}.
\end{eqnarray*}
is a squared output automaton for ${\cal T}$.
\end{proposition}
The squared automaton ${\cal A}^{sq}$ can be efficiently computed, see Algorithm~\ref{alg:sq_aut}.

\begin{algorithm}
\begin{multicols}{2}
\noindent
\begin{alltt}
ProcessState(\(F,\Delta,Q\sb{2},F\sb{2},\)
\hspace{2cm}\(\Delta\sb{2},start,end\))
@1\hspace{0.2cm}\(\tuple{p\sb{1},p\sb{2}}\leftarrow{Q\sb{2}}[start]\)
@2\hspace{0.2cm}if \(p\sb{1}\in{F}\) and \(p\sb{1}\in{F}\) then
@3\hspace{0.4cm}\(F\sb{2}\leftarrow{F}\sb{2}\cup\{\tuple{p\sb{1},p\sb{2}}\}\)
@4\hspace{0.2cm}for \(a\in\Sigma\) do
@5\hspace{0.4cm}for \(\tuple{p\sb{1},\tuple{a,m\sb{1}},q\sb{1}}\in\Delta\) do
@6\hspace{0.6cm}for \(\tuple{p\sb{2},\tuple{a,m\sb{2}},q\sb{2}}\in\Delta\) do
@7\hspace{0.7cm}if \(\tuple{q\sb{1},q\sb{2}}\not\in{Q}\sb{2}[0..{end-1}]\) then
@8\hspace{0.8cm}\(Q\sb{2}[end]\leftarrow\tuple{q\sb{1},q\sb{2}}\)
@9\hspace{0.8cm}\(end\leftarrow{end+1}\)
@10\hspace{0.7cm}\(\Delta\sb{2}\leftarrow{\tuple{p\sb{1},p\sb{2}},\tuple{m\sb{1},m\sb{2}},\tuple{q\sb{1},q\sb{2}}}\)
@11\hspace{0.2cm}return \(end\)
\end{alltt}
\break
\begin{alltt}
SquaredAutomaton(\({\cal{T}}\))
@1\hspace{0.2cm}\(\tuple{\Sigma\times{\cal{M}},Q,I,F,\Delta}\leftarrow{\cal T}\)
@2\hspace{0.2cm}\(Q\sb{2}\leftarrow\emptyset\), \(I\sb{2}\leftarrow\emptyset\)
@3\hspace{0.2cm}\(F\sb{2}\leftarrow\emptyset\), \(\Delta\sb{2}\leftarrow\emptyset\)
@4\hspace{0.2cm}\(start\leftarrow{0}\)
@5\hspace{0.2cm}\(end\leftarrow{0}\)
@6\hspace{0.2cm}for \(i\sb{1}\in{I}\) do
@7\hspace{0.4cm}for \(i\sb{2}\in{I}\) do
@8\hspace{0.6cm}\(Q\sb{2}[end]\leftarrow\tuple{i\sb{1},i\sb{2}}\)
@9\hspace{0.6cm}\(I\sb{2}\leftarrow{I}\sb{2}\cup\{\tuple{i\sb{1},i\sb{2}}\)
@10\hspace{0.6cm}\(end\leftarrow{end+1}\)
@11\hspace{0.2cm}while \(start<end\) do
@12\hspace{0.3cm}\(end\leftarrow{ProcessState(F,\Delta,Q\sb{2},F\sb{2},\Delta\sb{2},start,end)}\)
@13\hspace{0.3cm}\(start\leftarrow{start+1}\)
@14\hspace{0.2cm}done
@15\hspace{0.2cm}return \(\tuple{{\cal{M}}\times{\cal{M}},Q\sb{2},I\sb{2},F\sb{2},\Delta\sb{2}}\)
\end{alltt}
\end{multicols}
\caption{Computation of the squared automaton for a real-time transducer, ${\cal T}$.}\label{alg:sq_aut}
\end{algorithm}

\begin{definition}
Let ${\cal A}^2$ be a squared output automaton for the monoidal finite-state transducer ${\cal T}$. 
Let $\tuple{q_1,q_2}\in Q\times Q$ be on a successful path in ${\cal A}^2$, let 
$$\tuple{\tuple{i_1,i_2},\tuple{m_1,m_2},\tuple{q_1,q_2}}\in \Delta_2^*$$ 
for some initial state $\tuple{i_1,i_2}\in I\times I$. Then $\tuple{m_1,m_2}$ is called a {\em relevant pair} for $\tuple{q_1,q_2}$.  
\end{definition}
\begin{definition}\label{DefValuation}
Let 
$$
{\cal A}^2 = \tuple{{\cal M}\times {\cal M},Q\times Q,I\times I,F\times F,\Delta_2}
$$ 
be a squared output automaton for the monoidal finite-state transducer ${\cal T}$. 
A {\em valuation} of 
${\cal A}^2 = \tuple{{\cal M}\times {\cal M},Q\times Q,I\times I,F\times F,\Delta_2}$ 
is a pair of partial functions $\tuple{\rho,\nu}$, $\rho,\nu : Q^2 \rightarrow M^2$ such that for each state $\tuple{q_1,q_2}$ of the squared automaton $\rho(q_1,q_2)$ choses a relevant pair for $\tuple{q_1,q_2}$ if such a pair exists, and $\rho(q_1,q_2)$ is undefined otherwise. 
If $\rho(q_1,q_2)$ is defined and equalizable, then $\nu(q_1,q_2)$ returns a mge for $\rho(q_1,q_2)$, otherwise $\nu(q_1,q_2)$ is undefined. 
In addition we require that 
\begin{itemize}
\item
if $\tuple{q_1,q_2} \in I^2$, then $\rho(q_1,q_2)=\tuple{e,e}$,
\item
and if the chosen relevant pair is $\tuple{m,m}$ for some $m$, then $\nu$ returns $\tuple{e,e}$.
\end{itemize}
\noindent More formally:
$$\rho(q_1,q_2) = \left\{\begin{array}{ll} \tuple{e,e} & \mbox{if there exists a relevant pair for } \tuple{q_1,q_2} \in I^2 \\ \tuple{m_1,m_2} & \mbox{if there exists a relevant pair for } \tuple{q_1,q_2},\\
& \tuple{q_1,q_2} \not\in I^2 \mbox{ and }
    \tuple{m_1,m_2} \mbox{ is a relevant pair for } \tuple{q_1,q_2} \\ \neg ! & \mbox{otherwise} \end{array} \right. $$
$$\nu(q_1,q_2) = \left\{\begin{array}{ll} \tuple{e,e} & \mbox{if } !\rho(q_1,q_2) \mbox{ and } \rho(q_1,q_2)=\tuple{m,m}\\ \eta(\rho(q_1,q_2)) & \mbox{if } !\rho(q_1,q_2)\mbox{ and $\rho(q_1,q_2)$ equalizable}\\ \neg ! & \mbox{otherwise} \end{array} \right.
$$
\end{definition}

\begin{proposition}\label{prop:valuation}
We can effectively compute a valuation of ${\cal A}^2$.
\end{proposition}
\begin{proof}
Clearly $\tuple{e,e}$ is a relevant pair for each initial state $q \in I^2$. Hence, since the monoid operation in ${\cal M}$ is computable, we can compute $\rho$ by a straightforward traversal, say breadth-first search, of ${\cal A}^2$. Since $\eta$ is computable and the equality of elements of $M$ is decidable, we can also effectively obtain $\nu$. For technical details we refer to Algorithm~\ref{alg:valuation}. Note that in Algorithm~\ref{alg:valuation} the function {\tt Valuation} assumes that the states in $Q_2$ are ordered in a breadth-first order as returned by the function {\tt SquaredAutomaton} in Algorithm~\ref{alg:sq_aut}. The set $C$ is the set of co-accessible states.\qed
\end{proof}

\begin{algorithm}
\begin{multicols}{2}
\noindent
\begin{alltt}
Coaccessible(\({\cal{A}}\))
@1\hspace{0.2cm}\(\tuple{{\cal{M}},Q,I,F,\Delta}\leftarrow{\cal{A}}\)
@2\hspace{0.2cm}\(C\leftarrow\emptyset\) \(start\leftarrow{0}\) \(end\leftarrow{0}\)
@3\hspace{0.2cm}for \(p\in{F}\) do 
@4\hspace{0.4cm}\(C[end]\leftarrow{p}\)
@5\hspace{0.4cm}\(end\leftarrow{end+1}\)
@6\hspace{0.2cm}while \(start<end\) do
@7\hspace{0.4cm}\(p\leftarrow{C}[start]\)
@8\hspace{0.4cm}for \(\tuple{q,m,p}\in\Delta\) do
@9\hspace{0.6cm}if \(q\not\in{C}[0..{end-1}]\) then
@10\hspace{0.7cm}\(C[end]\leftarrow{q}\)
@11\hspace{0.7cm}\(end\leftarrow{end+1}\)
@12\hspace{0.5cm}fi
@13\hspace{0.3cm}done
@14\hspace{0.2cm}done
@15\hspace{0.2cm}return \(C\)
\end{alltt}
\break
\begin{alltt}
Valuation(\({\cal{A}},C,\eta\))
@1\hspace{0.2cm}\(\tuple{{\cal{M}}\times{\cal{M}},Q\sb{2},I\sb{2},F\sb{2},\Delta\sb{2}}\leftarrow{\cal{A}}\)
@2\hspace{0.2cm}\(\rho\leftarrow\emptyset\)
@3\hspace{0.2cm}for \(i=0\) to \({|Q\sb{2}|-1}\) do
@4\hspace{0.4cm}\(\tuple{p\sb{1},p\sb{2}}\leftarrow{Q\sb{2}}[i]\)
@5\hspace{0.4cm}if \(\tuple{p\sb{1},p\sb{2}}\in{C}\) then
@6\hspace{0.6cm}if \(\tuple{p\sb{1},p\sb{2}}\in{I}\sb{2}\)
@7\hspace{0.8cm}\(\rho(p\sb{1},p\sb{2})\leftarrow\tuple{e,e}\)
@8\hspace{0.6cm}for \({\tuple{\tuple{p\sb{1},p\sb{2}},\tuple{m\sb{1},m\sb{2}},\tuple{q\sb{1},q\sb{2}}}\in{\Delta}\sb{2}}\)
@9\hspace{0.8cm}if \(\tuple{q\sb{1},q\sb{2}}\in{C}\) and \(\rho(q\sb{1},q\sb{2})\) is not defined then
@10\hspace{0.9cm}\(\rho(q\sb{1},q\sb{2})\leftarrow\rho(p\sb{1},p\sb{2})\circ\tuple{m\sb{1},m\sb{2}}\)
@11\hspace{0.5cm}done
@12\hspace{0.3cm}fi
@13\hspace{0.2cm}done
@14\hspace{0.2cm}for \(\tuple{p\sb{1},p\sb{2}}\in{C}\) do
@15\hspace{0.4cm} \(\tuple{x\sb{1},x\sb{2}}\leftarrow\rho(p\sb{1},p\sb{2})\)
@16\hspace{0.4cm}if \({x\sb{1}=x\sb{2}}\) then
@17\hspace{0.6cm} \(\nu(p\sb{1},p\sb{2}=\tuple{e,e}\) 
@18\hspace{0.4cm}else
@19\hspace{0.6cm} \(\nu(p\sb{1},p\sb{2})\leftarrow\eta(x\sb{1},x\sb{2})\)
@20\hspace{0.2cm}done
@21\hspace{0.2cm}return \(\tuple{\rho,\nu}\)
\end{alltt}
\end{multicols}\caption{The valuation of the squared automaton, ${\cal A}^2$}\label{alg:valuation}
\end{algorithm}
The following lemma is crucial for the bimachine construction method described below. 
Part 1 says that if ${\cal T}$ is functional, then we may associate with each state of a squared automaton a {\em single} mge which acts as a 
mge for {\em all} relevant pairs of the state. 
\begin{lemma}\label{lemma:functional}
Let ${\cal T}$ be a trimmed monoidal finite-state transducer with output in the mge monoid ${\cal M}=\tuple{M,\circ,e}$.
Let ${\cal A}^2$ be a squared automaton for ${\cal T}$ with valuation $\tuple{\rho,\nu}$.
Then ${\cal T}$ is functional iff
\begin{enumerate}
\item 
for each relevant pair $\tuple{m_1,m_2}$ of a state $p\in Q^2$ always $\nu(p)$ is a mge for $\tuple{m_1,m_2}$
\item 
for each final and accessible state $f\in F\times F$ we have $\nu(f)=\tuple{e,e}$. 
\end{enumerate}
\end{lemma}

\begin{proof}
``$\Rightarrow$'': Let ${\cal T}$ be functional. 1. Let $\tuple{m_1,m_2}$ be a relevant pair for $p=\tuple{p_1,p_2}\in Q\times Q$. Therefore $p$ lies on a successful path in ${\cal A}^2$ and furthermore, there is an initial state $s=\tuple{s_1,s_2}\in I^2$ such that $\tuple{s,\tuple{m_1,m_2},p}\in \Delta_2^*$. Since $p$ lies on a successful path, 
there exists also a pair of final states $f=\tuple{f_1,f_2}$ and a tuple 
$\tuple{p,\tuple{k_1,k_2},f}\in \Delta_2^*$. Using that ${\cal A}^2$ is a squared automaton for ${\cal T}$, we conclude that there are words $u,v\in \Sigma^*$ such that $\tuple{s_i,\tuple{u,m_i},p_i}\in \Delta^*$ and $\tuple{p_i,\tuple{v,k_i},f_i}\in \Delta^*$ for $i=1,2$.
This shows that $\tuple{uv,m_ik_i}\in L({\cal T})$ for $i=1,2$ and since ${\cal T}$ is functional we get that
$m_1k_1=m_2k_2$ and $\tuple{k_1,k_2}$ is an equalizer for $\tuple{m_1,m_2}$. This implies that $\rho(p)$ and $\nu(p)$ are defined. Let $\rho(p)=\tuple{m'_1,m'_2}$. Since ${\cal T}$ is functional $\tuple{k_1,k_2}$ is an equalizer of $\tuple{m'_1,m'_2}$ and therefore $\tuple{k_1,k_2}$ is a joint equalizer of $\tuple{m_1,m_2}$ and $\tuple{m'_1,m'_2}$. From Lemma~\ref{LemmaJointEq} it follows that $\nu(p)$ is a mge for $\tuple{m_1,m_2}$ as well.\\
2. In the special case where $p\in F\times F$, we can take $k_1=k_2=e$. Hence, each relevant pair for $p$ has the form $(m,m)$ and by Definition~\ref{DefValuation} we have $\nu(p)=\tuple{e,e}$.

``$\Leftarrow$'': If ${\cal T}$ is {\em not} functional, then for some input string $w$ there are two distinct outputs $\tuple{m_1,m_2}$. Let 
$m_i$ be produced on a path with input label $w$ from an initial state $s_i$ to the final state $f_i$ ($i=1,2$). 
Since ${\cal A}^2$ is a squared automaton for ${\cal T}$, there is a path $\tuple{\tuple{s_1,s_2},\tuple{m_1,m_2},\tuple{f_1,f_2}}\in \Delta_2^*$.
Thus, $\tuple{m_1,m_2}$ is a relevant pair for $\tuple{f_1,f_2}$. Obviously, $\tuple{f_1,f_2}$ is final and accessible in ${\cal A}^2$. Since $m_1\neq m_2$ it follows that $\tuple{e,e}$ is {\em not} a mge for $\tuple{m_1,m_2}$.\qed 
\end{proof}
\begin{proposition}\label{PropCondEq}
Let ${\cal T}$ be a trimmed monoidal finite-state transducer with output in the mge monoid ${\cal M}=\tuple{M,\circ,e}$.
Let ${\cal A}^2$ be a squared automaton for ${\cal T}$ with valuation $\tuple{\rho,\nu}$.
Then the following two conditions are equivalent:
\begin{enumerate}
\item 
for each relevant pair $\tuple{m_1,m_2}$ of a state $p\in Q^2$ always $\nu(p)$ is a mge for $\tuple{m_1,m_2}$,
\item 
for all states $p,q \in Q^2$, such that $\rho(p)=\tuple{m_1,m_2}$ and $\rho(q)=\tuple{n_1,n_2}$ are defined and for each transition $\tuple{p,\tuple{u_1,u_2},q} \in \Delta_2$ always $\eta(m_1 u_1,m_2 u_2)$ is a mge of $\rho(q)$. 
\end{enumerate}
\end{proposition}
\begin{proof}
$(1 \Rightarrow 2)$ Since $\tuple{m_1 u_1, m_2 u_2}$ is a relevant pair of $q$ it follows from Condition~1 that $\nu(q)$ is a mge of $\tuple{m_1 u_1, m_2 u_2}$ and of $\rho(q)$. From Lemma~\ref{LemmaJointEq} it follows that $\eta(m_1 u_1,m_2 u_2)$ is a mge for $\rho(q)$ as well. 

$(2 \Rightarrow 1)$ We assume that Condition~2 holds. Let $\tuple{m_1,m_2}$ be an arbitrary relevant pair of a state $p=\tuple{p_1,p_2}\in Q^2$. Then $\tuple{m_1,m_2}$ is the label of a path with $n\in \bbbn$ transitions in ${\cal A}^2$. We prove Condition 1 by induction on $n$.\\
Let $n=0$. Then $\tuple{m_1,m_2}=\tuple{e,e}$ and $p \in I^2$, thus $\nu(p) = \tuple{e,e}$ is a mge for $\tuple{m_1,m_2}$.\\
Let us assume that for every relevant pair $\tuple{m_1,m_2}$ of a state $p \in Q^2$ that is a label of a path with  $n$ transitions in ${\cal A}^2$ we have that $\nu(p)$ is a mge for $\tuple{m_1,m_2}$.\\
Let $\tuple{n'_1,n'_2}$ be a relevant pair of a state $q\in Q^2$ that is the label of a path with $n+1$ transitions. Let the last transition of this path be $\tuple{p,\tuple{u_1,u_2},q} \in \Delta_2$. In this case there exists a pair $\tuple{m'_1,m'_2} \in M^2$ such that $n'_1=m'_1 u_1$, $n'_2=m'_2 u_2$ and $\tuple{m'_1,m'_2}$ is a relevant pair of $p$, and is a label of a path with $n$ transition. Let $\rho(p)=\tuple{m_1,m_2}$, $\rho(q)=\tuple{n_1,n_2}$ and $\eta(m_1 u_1, m_2 u_2)=\tuple{y'_1,y'_2}$. Then $m_1 u_1 y'_1 = m_2 u_2 y'_2$ and from Condition 2 we have that $n_1 y'_1 = n_2 y'_2$. Hence $\tuple{u_1y'_1,u_2y'_2}$ is an equalizer of $\tuple{m_1,m_2}$. From the induction hypothesis we have that $\nu(p) = \tuple{x_1,x_2}$ is a mge for $\tuple{m_1,m_2}$ and $\tuple{m'_1,m'_2}$ and therefore $\tuple{u_1y'_1,u_2y'_2}$ will be an equalizer of $\tuple{m'_1,m'_2}$. Hence $m'_1 u_1 y'_1 = m'_2 u_2 y'_2$ thus $\tuple{y'_1,y'_2}$ is an equalizer of $\tuple{n'_1,n'_2}$. Since $\tuple{n_1,n_2}$ and $\tuple{n'_1,n'_2}$ have a joint equalizer, $\nu(q)$ is a mge for $\tuple{n'_1,n'_2}$. \qed
\end{proof}

Now we can proceed with the functionality decision procedure, see also Algorithm~\ref{alg:functional}.
\begin{proposition}\label{lemma:func_test}
Assume that ${\cal M}$ is an effective mge monoid. Then the functionality of a real-time transducer ${\cal T}=\tuple{\Sigma,{\cal M},Q,I,\Delta,T}$ is decidable.
\end{proposition}
\begin{proof}
We proceed by first constructing the trimmed squared automaton ${\cal A}$ and a valuation for ${\cal A}$. Afterwards for each transition in the trimmed squared automaton ${\cal A}$ we check Condition 2 of Proposition~\ref{PropCondEq}. If the check fails, then ${\cal T}$ is not functional. Otherwise the transducer is functional iff for every final state $f \in F^2$ we have that $\nu(f)=\tuple{e,e}$.\qed
\end{proof}

\begin{algorithm}
\begin{multicols}{2}
\noindent
\begin{alltt}
EvaluatedSquaredAutomaton(\({\cal{T}},\eta\))
@1\hspace{0.2cm}\({\cal{A}}\leftarrow{SquaredAutomaton}({\cal{T}})\)
@2\hspace{0.2cm}\(C\leftarrow{Coaccessible}({\cal{A}})\)
@3\hspace{0.2cm}\(\tuple{\rho,\nu}\leftarrow{Valuation}({\cal{A}},C,\eta)\)
@4\hspace{0.2cm}return \({\cal{A}},C,\rho,\nu\)
\end{alltt}
\break
\begin{alltt}
TestFunctionality(\({\cal{T}},\eta\))
@1\hspace{0.2cm}\({\cal{A}},C,\rho,\nu\leftarrow{EvaluatedSqauredAutomaton}({\cal{T}},\eta)\)
@2\hspace{0.2cm}\(\tuple{{\cal{M}}\times{\cal{M}},Q\sb{2},I\sb{2},F\sb{2},\Delta\sb{2}}\leftarrow{\cal{A}}\)
@3\hspace{0.2cm}for \(\tuple{f\sb{1},f\sb{2}}\in{F}\sb{2}\) do
@4\hspace{0.4cm}if \(\nu(f\sb{1},f\sb{2})\neq\tuple{e,e}\) then
@5\hspace{0.6cm}return false
@6\hspace{0.2cm}for \({\tuple{\tuple{p\sb{1},p\sb{2}},\tuple{m\sb{1},m\sb{2}},\tuple{q\sb{1},q\sb{2}}}\in{\Delta}\sb{2}}\) do
@7\hspace{0.4cm}if \(\tuple{p\sb{1},p\sb{2}},\tuple{q\sb{1},q\sb{2}}\in{C}\) then
@8\hspace{0.6cm}if \(\nu(q\sb{1},q\sb{2})=\bot\) then
@9\hspace{0.8cm}return false
@10\hspace{0.5cm}else
@11\hspace{0.7cm}\(\tuple{x\sb{1},x\sb{2}}\leftarrow\rho(p\sb{1},p\sb{2})\)
@12\hspace{0.7cm}\(\tuple{y\sb{1},y\sb{2}}\leftarrow\nu(q\sb{1},q\sb{2})\)
@13\hspace{0.7cm}if \(x\sb{1}m\sb{1}y\sb{1}\neq{x}\sb{2}m\sb{2}y\sb{2}\) then
@14\hspace{0.9cm}return false
@15\hspace{0.2cm}done
@16\hspace{0.2cm}return true
\end{alltt}
\end{multicols}
\caption{Functionality test according to Proposition~\ref{PropCondEq}.}\label{alg:functional}
\end{algorithm}

\begin{remark}\label{complexity_issue}
The constraint on the original transducer to be a real-time entails the following complexity concern. In case that the original transducer is not real-time, we should perform a $\varepsilon$-closure procedure in advance. Whereas, this would not increase the number of states of the transducer, in the worst case this may lead to a squared increase of the number of transitions, $|\Delta| \sim |\Sigma||Q|^2$. This may harm the construction of the squared automaton and cause it to produce as many as $|\Sigma| |Q|^4$ transitions, whereas the original (not)real-time transducer may have had only $O(|\Sigma||Q|)$ transitions.  
\end{remark}

The following proposition suggests a solution for the concern raised by Remark~\ref{complexity_issue}. 
We use the notation introduced in Definition~\ref{DefExtExt}.

\begin{proposition}\label{prop:epsilon-sq_aut}
Let ${\cal T}=\tuple{\Sigma^\ast\times{\cal M},Q,I,\Delta,F}$ be a  monoidal finite-state transducer. Then
\begin{enumerate}
\item the monoidal automaton
\begin{eqnarray*}
{\cal A}_{e}^2 &=& \tuple{{\cal M}\times {\cal M},Q\times Q,I\times I,F\times F,\Delta^{e}_2}, \text{ where} \\
\Delta^{e}_2 &=& \{\tuple{\tuple{p_1,p_2},\tuple{m_1,m_2},\tuple{q_1,q_2}}\,|\, \exists a\in (\Sigma\cup \{\varepsilon\}) (\tuple{p_i,\tuple{a,m_i},q_i}\in \Delta^{ext} \text{ for } i=1,2)\}
\end{eqnarray*}
is a squared output automaton for ${\cal T}$.
\item the monoidal automaton:
\begin{eqnarray*}
{\cal A}^2 &=& \tuple{{\cal M}\times {\cal M},Q\times Q,I\times I,F\times F,\Delta_2}, \text{ where} \\
\Delta_2 & = & \{\tuple{\tuple{p_1,p_2},\tuple{m_1,m_2},\tuple{q_1,q_2}}\,|\, \exists a\in \Sigma (\tuple{p_i,\tuple{a,m_i},q_i}\in \Delta \text{ for } i=1,2)\}\\
&& \cup \{\tuple{\tuple{p_1,p_2},\tuple{e,m_2},\tuple{p_1,q_2}}\,|\, (\tuple{p_2,\tuple{\varepsilon,m_2},q_2}\in \Delta)\}\\
&& \cup \{\tuple{\tuple{p_1,p_2},\tuple{m_1,e},\tuple{q_1,p_2}}\,|\, (\tuple{p_1,\tuple{\varepsilon,m_1},q_1}\in \Delta)\}
\end{eqnarray*}
is a squared output automaton for ${\cal T}$.
\end{enumerate}
\end{proposition}
\begin{proof}
The first part is straightforward. For the second we prove $\Delta_2^*= (\Delta_2^{e})^*$.
Let $\tuple{p_i,\tuple{\varepsilon,m_i},q_i}\in \Delta^{e}$. Then 
\begin{eqnarray*}
\tuple{\tuple{p_1,p_2},\tuple{e,m_2},\tuple{p_1,q_2}} &\in& \Delta_2\\ 
\tuple{\tuple{p_1,q_2},\tuple{m_1,e},\tuple{q_1,q_2}} &\in& \Delta_2.
\end{eqnarray*}
Summing up we have that
$\tuple{\tuple{p_1,p_2},\tuple{m_1,m_2},\tuple{q_1,q_2}}\in \Delta_2^*$.
Therefore $\Delta_2^{e}\subseteq \Delta_2^*$.
Since, obviously $\Delta_2\subseteq \Delta_2^{e}$, we get $\Delta_2^*=(\Delta_2^{e})^*$. From the definition of a squared output automaton and Part 1 it follows that ${\cal A}^2$ is a squared output automaton for ${\cal T}$.\qed
\end{proof}
Thus applying the functional test to ${\cal A}^2$ yields a functional test algorithm for arbitrary transducer ${\cal T}$.

\begin{remark}\label{complexity_answer}
We can analyze the number of transitions in the modified squared automaton ${\cal A}^2$ as follows.
Let
\begin{eqnarray*}
\Delta_{a} &=&\{\tuple{p,\tuple{a,m},q}\in \Delta\} \text{ for } a\in \Sigma\cup \{\varepsilon\}
\end{eqnarray*}
Then the number of transitions in the modified ${\cal A}^2$ is bounded by:
\begin{equation*}
|\Delta_2| \le \sum_{a\in \Sigma} |\Delta_a|^2 + 2|Q| |\Delta_{\varepsilon}|.
\end{equation*}
Thus, we get at most a squared increase of the number of transitions. This is an especially desired bound
in the case where the transducer is constructed from a regular expression. In this case we can apply $\varepsilon$-constructions for union, concatenation and iteration. Specifically, we arrive at a transducer with a linear number of transitions and states in terms of the original regular expression. This shows that the upper bound for $|\Delta_2|$ in this case would be only the square of the original regular expression. This compares favourably to the worst case scenario described in Remark~\ref{complexity_issue}.
\end{remark}

\section{Bimachine construction based on the equalizer accumulation principle}\label{SecBimachineConstruction}

Let ${\cal T}=\tuple{\Sigma^*\times{\cal M},Q,I,\Delta,F}$ denote a functional transducer
with output in the effective mge monoid ${\cal M}=(M,\circ,e)$. Without loss of generality we assume that $\eta(e,e)=\tuple{e,e}$. Further we  assume that $\tuple{\varepsilon,e} \in L({\cal T})$ (cf. Remark~\ref{real-time2}).

In this section we show how to construct a bimachine with at most $2^{|Q|+1}$ states that is equivalent to ${\cal T}$.
Specifically, we prove:
\begin{theorem}\label{th:contribution_real_time}
If ${\cal T}$ is a functional real-time transducer, then there is a bimachine ${\cal B}=\tuple{{\cal M},{\cal A}_L,{\cal A}_R,\psi}$ such that ${\cal A}_L$ has at most $2^{|Q|}$ and ${\cal A}_R$ has at most $2^{|Q|}$ states and $O_{\cal B}  = O_{\cal T}$.
\end{theorem}
\begin{theorem}\label{th:contribution}
If ${\cal T}$ is an arbitrary functional transducer, then there is a bimachine ${\cal B}=\tuple{{\cal M},{\cal A}_L,{\cal A}_R,\psi}$ such that ${\cal A}_L$ has at most $2^{|Q|}$ and ${\cal A}_R$ has at most $2^{|Q|}$ states and $O_{\cal B}  = O_{\cal T}$.
\end{theorem}
In both cases, the left deterministic automaton ${\cal A}_L$ results from the determinization of the underlying automaton for ${\cal T}$, whereas the right deterministic automaton ${\cal A}_R$ results from the determinization of the reversed underlying automaton for ${\cal T}$. Of course, in the second case we adopt $\varepsilon$-power-set construction. The subtle part of the 
construction is the definition of the output function, $\psi$. It is in this part where the squared automaton, ${\cal A}^2$, and the notion of mge's come into play.

We present our algorithm stepwise. First, we consider a special case and use it to informally describe our intuition and motivation for the construction. Sections~\ref{bimachine:construction} and~\ref{bimachine:correctness} present the proof of Theorem~\ref{th:contribution_real_time}. In Section~\ref{bimachine:construction} we give the formal construction and we prove its correctness in Section~\ref{bimachine:correctness}. In Section~\ref{bimachine:general} we show that very subtle and natural amendments of our
construction yield the proof of Theorem~\ref{th:contribution}.
\subsection{High-level description}
To illustrate the idea of our algorithm, we consider an example where ${\cal T}$ is a functional real-time transducer with outputs in ${\mathbb R}_{\ge 0}=({\mathbb R}_{\ge 0},+,0)$.

Let $u=a_1\dots a_n$ be in the domain of ${\cal T}$ and consider
\begin{figure}
\includegraphics[width=0.9\textwidth]{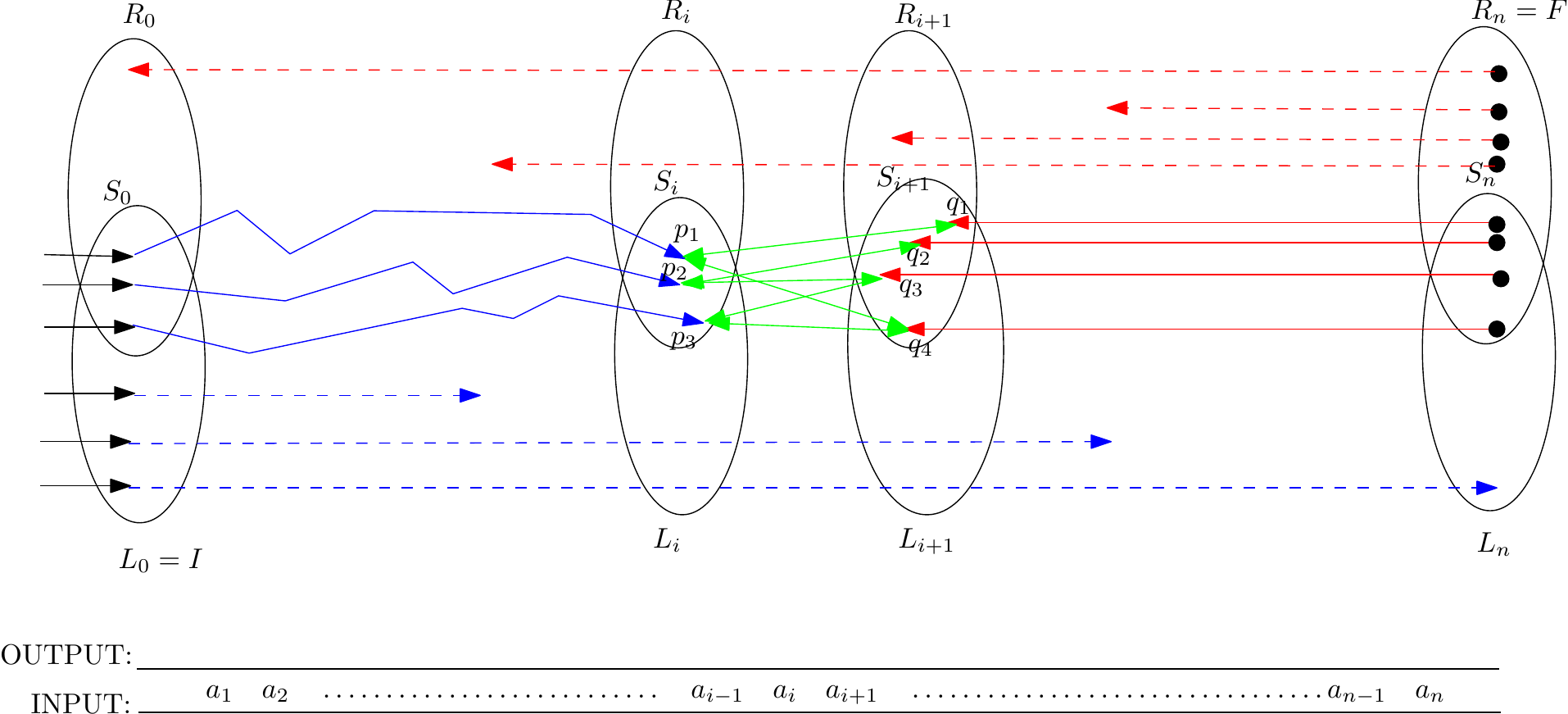}
\caption{}\label{fig:basic}
\end{figure}
the ordinary, i.e. left-to-right, traversal of $u$ in ${\cal T}$, see Figure~\ref{fig:basic}. It starts from the set of initial states, $L_0=I$, and gradually extends $L_i$ to $L_{i+1}$ by following all the outgoing transitions from $L_i$ labelled with input character $a_{i+1}$. Similarly, the reverse, i.e. right-to-left, traversal of $u$ in ${\cal T}$ starts from $R_n=F$ and stepwise turns back from $R_{i+1}$ to $R_i$ by following all the incoming transitions in $R_{i+1}$ that are labelled with $a_{i+1}$. Jointly, $L_i$ 
and $R_i$ express that the successful paths traverse exactly the states $S_i=L_i\cap R_i$ before scanning the character $a_{i+1}$. At 
the time step when $a_{i+1}$ is scanned, all these paths depart from $S_i$, follow a transition labelled $a_{i+1}$, and arrive in $S_{i
+1}=L_{i+1}\cap R_{i+1}$, see Figure~\ref{fig:basic}.

Let us assume that $S_i=\{p_1,p_2,p_3\}$ and each of the states, $p_j$, is reached by a path $\pi_j$ for $j\in \{1,2,3\}$ labeled with $a_1\dots a_i$ on the input. Assume that the path $\pi_j$ produces an output $c_j=cost(\pi_j)$. Let $M_i=\max(c_1,c_2,c_3)$ be the maximum over all the outputs $c_j$. We shall maintain the following invariant. By time step $i$, we will emit exactly the maximum, $M_i$. 

Our construction relies on the following observation.
Since each state $p_j \in S_i$ lies on a successful path labeled with $a_1\ldots a_n$ for $j\in \{1,2,3\}$ we can consider a path $P_j$ with input label  $a_{i+1}\ldots a_n$ starting from $p_j$ and terminating in a final state. Since the transducer ${\cal T}$ is functional, the outputs produced along these paths must coincide, i.e.:
\begin{equation*}
cost(\pi_j) +cost(P_j) =cost(\pi_k)+cost(P_k) \text{ for all }j,k\in \{1,2,3\}.\quad (\dagger)
\end{equation*}
This implies that each state $p_j\in S_i$ can (precompute and) store its imbalance:
$$v_{S_i}(p_j)=\max(c_1,c_2,c_3)-c_j=M_i-c_j.$$
Note that the imbalances depend only on the set $S_i$ and do not depend on the specific input $a_1 \ldots a_i$. 
This follows from the equations ($\dagger$) and the fact that the paths $P_j$ do not depend on $a_1 \ldots a_i$.
A similar argument applies to the set $S_{i+1}$. In our example $S_{i+1}=\{q_1,q_2,q_3,q_4\}$. Thus, if four paths, $\pi_1',\pi_2',\pi_3',\pi_4'$ reach $q_1,q_2,q_3,q_4$, respectively, with the same input, then for their costs, $c_j'=cost(\pi_j')$ it holds that:
$$v_{S_{i+1}}(q_j)=M_{i+1}-c'_j.$$
Now we select an arbitrary transition $t_i=\tuple{p,\tuple{a_{i+1},m_{i+1}},q}\in \Delta$ with $p=p_j\in S_i$ and $q=q_k \in S_{i+1}$. The path following $\pi_j$  through $t_i$ has cost:
\begin{equation*}
c_k'=c_j+m_{i+1} =M_i - (M_i-c_j)+m_{i+1}=M_i - v_{S_i}(p_j) +m_{i+1}
\end{equation*} 
On the other hand:
\begin{equation*}
M_{i+1}=c_k'+v_{S_{i+1}}(q_k)=M_i - v_{S_i}(p_j)+m_{i+1}+v_{S_{i+1}}(q_k).
\end{equation*}
Hence:
\begin{equation*}
M_{i+1}-M_{i}=m_{i+1}+v_{S_{i+1}}(q_k) - v_{S_i}(p_j).
\end{equation*}
The right hand side of this expression does not depend on the specific transition we take from $S_i$ to $S_{i+1}$ and can be expressed only locally in terms of the structure of ${\cal T}$ and the precomputed imbalances, $v_{S_i}(p_j)$ and $v_{S_{i+1}}(q_k)$ .

Using the above observation we construct the bimachine by first computing the left and right deterministic automata applying the standard determinization procedure to the underlying and the reversed underlying automaton. Then the output function is defined by setting $\psi(L_i,a,R_{i+1}) = M_{i+1}-M_{i}$.

\subsection{Formal construction}\label{bimachine:construction}
As before let ${\cal T}$ be a functional real-time transducer with outputs in the effective mge monoid 
${\cal M}=(M,\circ,e)$. Let $\eta$ denote a function that computes a mge for each equalizable pair of $M$.   
The construction of the bimachine ${\cal B}$ for ${\cal T}$ uses five steps. 

\noindent{\bf Step 1.\ } We compute 
the {\bf squared output automaton} 
$$
{\cal A}^2 = \tuple{{\cal M}\times {\cal M},Q \times Q,I \times I,\Delta_2,F \times F}\\
$$
for ${\cal T}$ and a valuation $\tuple{\rho,\nu}$ for ${\cal A}^2$ as described in 
Proposition~\ref{prop:valuation}, see the left part of Algorithm~\ref{alg:functional}.

\noindent{\bf Step 2.\ } We compute the left and right deterministic automata for the bimachine, see Algorithm~\ref{alg:step2}.
The {\bf left deterministic automaton} is defined as the result 
$${\cal A}_L = \tuple{\Sigma, Q_L, s_L, Q_L, \delta_L}$$
when applying the standard determinization procedure\footnote{By the ``standard'' construction we mean the one that generates only the accessible states in the deterministic automaton.} to the underlying automaton of ${\cal T}$. This means that $Q_L \subseteq 2^Q$, $s_L = I$ and the transition function is defined as
$\delta_L(L,a) := \set{q'}{\exists q \in L, m \in M : \tuple{q,\tuple{a,m},q'} \in \Delta}$.
The {\bf right deterministic automaton} is defined as the result 
$${\cal A}_R = \tuple{\Sigma, Q_R, s_R, Q_R, \delta_R}$$ when applying the standard 
determinization procedure to the reversed underlying automaton of  ${\cal T}$. We have $Q_R \subseteq 2^Q$, $s_R = F$ and
the transition function is defined as $\delta_R(R,a) := \set{q}{\exists q' \in R, m \in M : \tuple{q,\tuple{a,m},q'} \in \Delta}$.

\noindent{\bf Step 3 (mge accumulation, see also Algorithm~\ref{alg:step3}).\ } 
Let 
$$
{\cal S} =\{L\cap R\neq \emptyset \,|\, L\in Q_L, R\in Q_R\}
$$ 
%
%
For each $S\in {\cal S}$ we define a (partial) function $\phi_{S}:{\cal S}\rightarrow M$ that represents a mge for $S$: 
fix any enumeration $\langle p_1,p_2,\dots,p_k\rangle$ of the elements of $S$. Then
\begin{equation*}
\phi_S(p_i)= \gamma_i^{(k)}(\nu(p_1,p_2),\dots,\nu(p_{k-1},p_k))
\end{equation*}
where $\gamma^{(k)}:(M\times M)^{k-1}\rightarrow M^{k}$ is the mge accumulation function introduced in  Corollary~\ref{CorollaryCombinationPairwisecmges2}.  

\noindent{\bf Step 4.\ } 
With these notions we now define the {\bf output function} $\psi:Q_L\times \Sigma\times Q_R\rightarrow {\cal M}$, see Algorithm~\ref{alg:step4}. To define $\psi(L,a,R')$ we proceed as follows: 
\begin{enumerate}
\item Let $L' :=\delta_L(L,a)$ and $R := \delta_R(R',a)$. If $L\cap R=\emptyset$, 
      then $\psi(L,a,R')$ is undefined. Otherwise
\item Let $S=L\cap R$ and $S'=L'\cap R'$. 
\item\label{step:transition} Let $t=\tuple{p,\tuple{a,m},p'}\in \Delta$ be arbitrary with $p\in S$ and $p'\in S'$. 
\item We define $\psi(L,a,R'):=c$, where $c$ is the unique solution of the equation $\phi_S(p) \circ c = m\circ \phi_{S'}(p')$ (cf. Corollary~\ref{CorSolvability}).
\end{enumerate}

\noindent{\bf Step 5.\ } Finally, we define ${\cal B} :=\tuple{{\cal M},{\cal A}_L,{\cal A}_R,\psi}$.

\begin{algorithm}
\begin{multicols}{2}
\noindent
\begin{alltt}
SetTrans\(({\Delta},P,a)\)
@1\hspace{0.2cm}return \(\{q\,|\,\exists{p}\in{P}(\tuple{p,a,q}\in\Delta)\}\)

Determinize\(({\cal{A}})\)
@1\hspace{0.20cm}\(\tuple{\Sigma,Q,I,F,\Delta}\leftarrow{\cal{A}}\)
@2\hspace{0.20cm}\(Q\sp{(-1)}\sb{D}\leftarrow\emptyset\); \(Q\sp{(0)}\sb{D}\leftarrow\{I\}\) 
@3\hspace{0.20cm}\(\delta\sb{D}\leftarrow\emptyset\); \(i\leftarrow{0}\);
@3\hspace{0.20cm}while \(Q\sb{D}\sp{(i)}\neq{Q}\sb{D}\sp{(i-1)}\) do
@4\hspace{0.40cm}\(Q\sb{D}\sp{(i+1)}\leftarrow{Q}\sb{D}\sp{(i)}\)
@5\hspace{0.40cm}for \(P\in{Q}\sb{D}\sp{(i)}\setminus{Q}\sb{D}\sp{(i-1)}\) do
@6\hspace{0.60cm}for \(a\in\Sigma\) do
@7\hspace{0.90cm}\(\delta\sb{D}(P,a)\leftarrow{SetTrans(\Delta,P,a)}\)
@8\hspace{0.90cm}\(Q\sb{D}\sp{(i+1)}\leftarrow{Q}\sb{D}\sp{(i+1)}\cup\{\delta\sb{D}(P,a)\}\)
@9\hspace{0.40cm}\(i\leftarrow{i+1}\)
@10\hspace{0.20cm}return \(\tuple{\Sigma,Q\sb{D}\sp{(i)},\{I\},Q\sb{D}\sp{(i)},\delta\sb{D}}\)
\end{alltt}
\break
\begin{alltt}
Project\(({\Delta})\)
@1\hspace{0.2cm}return\(\{\tuple{p,a,q}|\exists{m}(\langle{p,\langle{a,m}\rangle,q}\rangle{\in\Delta})\}\)

Reverse\(({\Delta})\)
@1\hspace{0.2cm}return \(\{\tuple{q,a,p}\,|\,\tuple{p,a,q}\in\Delta\}\)

BimachineAutomata(\({\cal{T}}\))
@1\hspace{0.2cm}\(\tuple{\Sigma\times{\cal{M}},Q,I,F,\Delta}\leftarrow{\cal{T}}\)
@2\hspace{0.2cm}\(\Delta\sb{\Sigma}\leftarrow{Project(\Delta)}\)
@3\hspace{0.2cm}\(\Delta\sp{rev}\sb{\Sigma}\leftarrow{Reverse(\Delta\sb{\Sigma})}\)
@4\hspace{0.2cm}\({\cal{A}}\leftarrow\tuple{\Sigma,Q,I,F,\Delta\sb{\Sigma}}\)
@5\hspace{0.2cm}\({\cal{A}}\sp{rev}\leftarrow\tuple{\Sigma,Q,F,I,\Delta\sp{rev}\sb{\Sigma}}\)
@6\hspace{0.2cm}\({\cal{A}}\sb{L}\leftarrow{Determinize}({\cal{A}})\)
@7\hspace{0.2cm}\({\cal{A}}\sb{R}\leftarrow{Determinize}({\cal{A}}\sp{rev})\)
@8\hspace{0.2cm}return \(\tuple{{\cal{A}}\sb{L},{\cal{A}}\sb{R}}\)
\end{alltt}
\end{multicols}
\caption{Step 2 of the construction. Determinize is a standard power-set determinization algorithm for automata with no $\varepsilon$-transitions. BimachineAutomata constructs the left and right automaton as described in Step 2.}\label{alg:step2}
\end{algorithm}


\begin{algorithm}
\begin{multicols}{2}
\fontsize{9}{9}
{
\noindent
\begin{alltt}
n-MGE(\(\eta,n,x\))
//\(x[1..n-1]\) is an array of pairs 
//\(x[i]=\tuple{x\sb{1}[i],x\sb{2}[i]}\)
@1\hspace{0.2cm}if \(n=1\) return \(e\)
@2\hspace{0.2cm}if \(n=2\) return \(x[1]\)
@3\hspace{0.2cm}\(\tuple{z\sb{1},z\sb{2},\dots,z\sb{n-1}}\leftarrow{\text{n-MGE}}(\eta,{n-1},x)\)
@4\hspace{0.2cm}for \(i=1\) to \(n-1\)
@5\hspace{0.4cm}\(y\sb{i}\leftarrow{z\sb{i}}\circ\eta\sb{1}(z\sb{n-1},x\sb{1}[n-1])\)
@6\hspace{0.2cm}\(y\sb{n}\leftarrow{x}\sb{2}[n-1]\circ\eta\sb{2}(z\sb{n-1},x\sb{1}[n-1])\)
@6\hspace{0.2cm}return \(\tuple{y\sb{1},y\sb{2},\dots,y\sb{n}}\)
\end{alltt}

\begin{alltt}
SyntacticMGE(\(S,\eta,\nu\))
@1\hspace{0.2cm}\(S\leftarrow\tuple{p\sb{1},\dots,p\sb{n}}\)
@2\hspace{0.2cm}if \(n=1\) return \(e\)
@3\hspace{0.2cm}for \(i=1\) to \(n-1\)
@4\hspace{0.4cm}\(x[i]\leftarrow\nu(p\sb{i},p\sb{i+1})\)
@5\hspace{0.2cm}return \(\text{n-MGE}(\eta,n,x)\)
\end{alltt}

\break

\begin{alltt}
SetsMGE(\(Q\sb{L},Q\sb{R},\eta,\nu\))
@1\hspace{0.2cm}\({\cal{S}}\leftarrow\emptyset\), \({\Phi}\leftarrow\emptyset\)
@2\hspace{0.2cm}for \(L\in{Q\sb{L}}\) do
@3\hspace{0.4cm}for \(R\in{Q\sb{R}}\) do
@4\hspace{0.6cm}\(S\leftarrow{L}\cap{R}\)
@5\hspace{0.6cm}if \(S\neq\emptyset\) and \(S\not\in{\cal{S}}\) then
@6\hspace{0.8cm}\(\{p\sb{1},\dots,p\sb{n}\}\leftarrow{S}\)
@7\hspace{0.8cm}\(S'\leftarrow\tuple{p\sb{1},\dots,p\sb{n}}\)
@8\hspace{0.8cm}\(\Phi(S)\leftarrow{SyntacticMGE(S,\eta,\nu)}\)
@9\hspace{0.6cm}\({\cal{S}}\leftarrow{\cal{S}}\cup\{S\}\)
@10\hspace{0.4cm}fi
@11\hspace{0.3cm}done
@12\hspace{0.2cm}done
@13\hspace{0.2cm}return \(\tuple{{\cal{S}},\Phi\}}\)
\end{alltt}
}
\end{multicols}
\caption{On the left n-MGE computes a mge for an $n$-tuple of equalizable monoidal elements $\tuple{m_1,\dots,m_n}$. Its implementation follows Corollary~\ref{CorollaryCombinationPairwisecmges2}. The function SyntacticMGE follows the description of $\phi_S$ in Step 3. On the right, SetsMGE takes as input the set of sets of states, $Q_L$ and $Q_R$, and implements Step 3 of the construction. The selector function, $\sigma$, is fixed to take the first element of some arbitrary ordering of the non-empty set $S$, and $\Phi(S)$ corresponds to $\phi_S$ in the description.}\label{alg:step3}
\end{algorithm}

\begin{algorithm}
\begin{multicols}{2}

\fontsize{9}{9}
\noindent
\begin{alltt}
{
Output(\(L,R',a,\delta\sb{L},\delta\sb{R},\Delta,\Phi,\eta,\sp{-1})\)
@1\hspace{0.2cm}\(L'\leftarrow{\delta\sb{L}(L,a)}\)
@2\hspace{0.2cm}\(R\leftarrow{\delta\sb{R}(R',a)}\)
@3\hspace{0.2cm}\(S\leftarrow{L}\cap{R}\)
@4\hspace{0.2cm}if \(S=\emptyset\) return \(\bot\)
@5\hspace{0.2cm}\(S'\leftarrow{L'}\cap{R'}\)
@7\hspace{0.2cm}let \(p\in{S},p'\in{S'}:\tuple{p,\tuple{a,m},p'}\in{\Delta}\)
@8\hspace{0.2cm}\(\phi\sb{S}\leftarrow\Phi(S)\)
@9\hspace{0.2cm}\(\phi\sb{S'}\leftarrow\Phi(S')\)
@10\hspace{0.1cm}\(\tuple{x\sb{1},x\sb{2}}\leftarrow\eta(\phi\sb{S}(p),m\phi\sb{S'}(p'))\)
@11\hspace{0.1cm}return \(x\sb{1}x\sb{2}\sp{-1}\)
}
\end{alltt}

\break

\begin{alltt}
BimachineConstruction(\({\cal{T}},\eta,\sp{-1}\))
@1\hspace{0.2cm}\(\tuple{{\cal{A}},C,\rho,\nu}\leftarrow{EvaluatedSquaredAutomaton({\cal{T}},\eta)}\)
@2\hspace{0.2cm}\(\tuple{{\cal{A}}\sb{L},{{\cal{A}}\sb{R}}}\leftarrow{BimachineAutomata}(\cal{T})\)
@3\hspace{0.2cm}\(\tuple{\Sigma,Q\sb{L},s\sb{L},Q\sb{L},\delta\sb{L}}\leftarrow{{\cal{A}}\sb{L}}\)
@4\hspace{0.2cm}\(\tuple{\Sigma,Q\sb{R},s\sb{R},Q\sb{R},\delta\sb{R}}\leftarrow{{\cal{A}}\sb{R}}\)
@5\hspace{0.2cm}\(\tuple{{\cal{S}},\Phi}\leftarrow{SetsMGE}(Q\sb{L},Q\sb{R},\eta,\nu)\)
@6\hspace{0.2cm}\(\psi\leftarrow\emptyset\)
@7\hspace{0.2cm}\(\tuple{\Sigma\times{\cal{M}},Q,I,F,\Delta}\leftarrow{\cal{T}}\)
@8\hspace{0.2cm}for \(L\in{Q}\sb{L}\) do
@9\hspace{0.3cm}for \(R\in{Q}\sb{R}\) do
@10\hspace{0.3cm}for \(a\in{\Sigma}\) do
@11\hspace{0.6cm}\(x\leftarrow{Output}(L,R,a,\delta\sb{L},\delta\sb{R},\)
\hspace{3cm}\(\Delta,\Phi,\eta,\sp{-1})\)
@12\hspace{0.6cm}if \(x\neq\bot\) then
@13\hspace{0.8cm}\(\psi(L,a,R)\leftarrow{x}\)
@14\hspace{0.4cm}done
@15\hspace{0.3cm}done
@16\hspace{0.2cm}done
@17\hspace{0.2cm}return \(\tuple{{\cal{M}},{\cal{A}}\sb{L},{\cal{A}}\sb{R},\psi}\)
\end{alltt} 
\end{multicols}
\caption{Step 4 of the construction. The function Output computes the $\psi(L,a,R')$ as described in Step 4 and
BimachineConstruction applies Steps 1-4 to compile the ultimate bimachine.}\label{alg:step4}
\end{algorithm}

\subsection{Correctness for real-time functional transducers}\label{bimachine:correctness}
We first show that Step 4 unambiguously defines the domain and the values of the output the function, $\psi$. 
Formally, we have the following proposition. 
%
\begin{proposition}\label{prop:well-defined}
Let ${\cal A}^2$, ${\cal A}_L$, and ${\cal A}_R$ be defined as in Steps~1 and~2 above. Assume that $L,L'\in Q_L$, $R,R'\in Q_R$ and $a\in \Sigma$ satisfy
\begin{equation*}
L'=\delta_L(L,a) \text{ and } R=\delta_R(R',a)
\end{equation*}
and let $S=L\cap R$ and $S'=L'\cap R'$. If $S \not= \emptyset$, then:
\begin{enumerate}
\item  
for each state $p\in S$ there exists a transition $t=\tuple{p,\tuple{a,m},p'}$ with $p'\in S'$.
\item $c=\psi(L,a,R')$ is well-defined and for every transition $\tau=\tuple{p_\tau,\tuple{a,m_\tau},p'_\tau}$ with $p_\tau \in S$ and $p'_\tau \in S'$ it holds:
\begin{equation*}
\phi_S(p_\tau) c=m_\tau \phi_{S'}(p'_\tau) .
\end{equation*}
\end{enumerate}
\end{proposition}
\begin{proof}
1. Let $p\in S$. Since $S=L\cap R$, and $R=\delta_R(R',a)$ by the determinization construction of ${\cal A}_R$ there is a transition  $t=\tuple{p,\tuple{a,m},p'}\in \Delta$ with $p'\in R'$. Since $p\in L$ from the determinization construction of ${\cal A}_L$ we get that $p'\in \delta_L(L,a)=L'$. Therefore $p'\in L'\cap R'=S'$ as required.

2. Let $u$ be such that $L=\delta_L^\ast(s_L,u)$. Let $S=\{p_1,\ldots,p_k\}$. Then there exist $s_1,\ldots,s_k \in I$ and $n_1,\ldots,n_k \in M$ such that $\tuple{s_j,\tuple{u,n_j},p_j} \in \Delta^\ast$ for each $1 \le j \le k$. Using Point 1, for each $j\le k$ we fix a transition $t_j=\tuple{p_j,\tuple{a,m_j},p'_j}\in \Delta$ such that $p'_j \in S'$. Let $\tau=\tuple{p_\tau,\tuple{a,m_\tau},p'_\tau}$ with $p_\tau \in S$ and $p'_\tau \in S'$  be arbitrary. In particular, $p_\tau=p_{i}$ for some $i\le k$.

Let $v$ be such that $R'=\delta_R^\ast(s_R,v^{rev})$. Then there exist $f_1,\ldots,f_k,f_\tau \in F$ and $n'_1,\ldots,n'_k,n'_\tau \in M$ such that $\tuple{p'_j,\tuple{ v,n'_j},f_j} \in \Delta^\ast$ for $1 \le j \le k$ and $\tuple{p'_\tau,\tuple{v,n'_\tau},f_\tau}\in \Delta^*$. Since the transducer is functional we have that $n_1 m_1 n'_1 = \ldots = n_k m_k n'_k=n_i m_\tau n'_\tau$. Therefore $\tuple{n_1 m_1,\ldots,n_k m _k,n_i m_\tau}$ and $\tuple{n_1,\ldots,n_i}$ are equalizable. From Lemma~\ref{lemma:functional} we get that $\nu(p_j,p_{j+1})$ is a mge for $\tuple{n_j,n_{j+1}}$. Thus, if $x_j=\phi_{S}(p_j)$ for $j\le k$ by Corollary~\ref{CorollaryCombinationPairwisecmges2} we get that $\phi_{S}=\tuple{x_1,\dots,x_k}$ is a mge of $\tuple{n_1,\ldots,n_k}$. Similarly, if $y_j=\phi_{S'}(p_j')$ and $y_\tau=\phi_{S'}(p'_\tau)$ we get
that $\tuple{y_1,\dots,y_k,y_\tau}$ is an equalizer for $\tuple{n_1m_1,\dots,n_km_k,n_i m_\tau}$. Hence, $\tuple{m_1y_1,\dots,m_ky_k,m_\tau y_\tau}$ is an equalizer for $\tuple{n_1,\ldots,n_k,n_i}$. Therefore there exists $c\in M$ such that $x_j c =m_j y_j$ for $j\le k$ and
further $x_i c=m_\tau y_\tau$. Recalling that $x_j=\phi_{S}(p_j)$, $y_j=\phi_{S'}(p_j')$ and $y_\tau=\phi_{S'}(p'_\tau)$ we get:
\begin{eqnarray*}
\phi_S(p_j) c = m_j \phi_{S'}(p'_j) \text{ for } j \le k \text{ and } \phi_S(p_\tau) c =m_\tau \phi_{S'}(p'_\tau).
\end{eqnarray*}
By the discussion above we know that this system has a solution. Hence each of the equations has an unique solution (cf. Lemma~\ref{lemmaSolvability}) and therefore the solution is uniquely defined by any of the equations. Hence, $\psi(L,a,R')$ is well defined.
  \qed
\end{proof}
\begin{theorem}\label{th:groups}
 Let $\mathcal{T}$ and ${\cal B}$ be as above and $u=a_1\dots a_n\in \Sigma^+$. Then:
 \begin{enumerate}
 \item if $\tuple{u,m}\in L({\cal T})$, then $O_{\cal B}(u)=m$.
 \item $O_{\cal B}=O_{\cal T}$.
 \end{enumerate}
 \end{theorem}
 \begin{proof}
 \begin{enumerate}
 \item Since $\tuple{u,m}\in L({\cal T})$ there is a successful path:
 \begin{equation*}
 \tuple{p_0,\tuple{a_1,m_1},p_1}\dots \tuple{p_{n-1},\tuple{a_n,m_n},p_n} \text{ in } {\cal T}.
 \end{equation*}
 In particular, $p_0\in I$, $p_n\in F$, and $m=m_1\circ m_2\dots \circ m_n$.
 Let $L_i=\delta_L^*(s_L,a_1\dots a_{i})$ and $R_i=\delta^*_R(s_R,a_n\dots a_{i+1})$.
By the power-set determinization construction of ${\cal A}_L$ and ${\cal A}_R$ it follows
that $p_i\in L_i$ and $p_i\in R_i$. Let $S_i=L_i\cap R_i$. Hence $S_i\neq \emptyset$ for
all $i\le n$. By Proposition~\ref{prop:well-defined}, Point 2, it follows that:
\begin{eqnarray*}
c_{i+1}&=&\psi(L_i,a_{i+1},R_{i+1}) \text{ is well-defined and } \\
\phi_{S_i}(p_i) c_{i+1}& =& m_{i+1} \phi_{S_{i+1}}(p_{i+1}).
\end{eqnarray*}
Now a straightforward computation shows that:
$$
\phi_{S_0}(p_0) c_1 c_2 \ldots c_n = m_1 \phi_{S_1}(p_1) c_2 c_3 \ldots c_n = \cdots = m_1 m_2\dots m_n \phi_{S_n}(p_n).
$$
Since $S_0\subseteq I$, we have that for all $p_0',p_0''\in S_0$, $\nu(p_0',p_0'')=\tuple{e,e}$. Therefore, since
$\eta(e,e)=\tuple{e,e}$ by Remark~\ref{pure_unity_case}, we get $\phi_{S_0}(p_0)=e$ for all $p_0\in S_0$. Similarly, since $S_n\subseteq F$ all the states in $S_n$ are final. By the functionality test we have that $\nu(p_n',p_n'')=\tuple{e,e}$ for all $p_n',p_n''\in S_n$. Again, since $\eta(e,e)=\tuple{e,e}$, by Remark~\ref{pure_unity_case} $\phi_{S_n}(p_n)=e$.
Therefore:
\begin{equation*}
O_{\cal B}(u)=c_1 \ldots c_n=m_1 \ldots m_n=m.
\end{equation*}
\item First, $O_{\cal B}(\varepsilon)=\varepsilon=O_{\cal T}(\varepsilon)$ by our convention in the beginning of the section. By the first part, we have that $O_{\cal T}\subseteq O_{\cal B}$. Conversely, if $u\in \Dom({\cal B})$ and $u\in \Sigma^+$ it follows that $\delta^*_L(s_L,u)\cap s_R\neq \emptyset$. By the power-set construction of ${\cal A}_L$ and since $s_R=F$ this means that there is a successful path labelled by $u$ in ${\cal T}$. Hence,
$u\in \Dom({\cal T})$. 
 \end{enumerate}\qed
 \end{proof}

\begin{proof}[of Theorem~\ref{th:contribution_real_time}]
Since by Step 2 of the construction, ${\cal A}_L$ and ${\cal A}_R$ have at most $2^{|Q|}$ states,
the result follows by Theorem~\ref{th:groups}. \qed
\end{proof}

\subsection{Non-real-time functional transducers}\label{bimachine:general}
Next, we turn our attention to the general case, where the functional transducers is not necessarily real-time.
The main issue to be addressed are the $\tuple{\varepsilon,m}$-transitions. Nevertheless, very natural amendments
to the construction from Section~\ref{SecBimachineConstruction} yield the result of Theorem~\ref{th:contribution}.
We start with the following verbatim modification:

\noindent{\bf Step 1-$\varepsilon$.\ } See Step~1. For the construction of a squared automaton in this case
we apply Propostion~\ref{prop:epsilon-sq_aut}.

\noindent{\bf Step 2-$\varepsilon$.\ } See Step~2. The only difference here is that we apply a $\varepsilon$-power-set determinization to obtain:
\begin{eqnarray*}
{\cal A}_L = \tuple{\Sigma, Q_L, s_L, Q_L, \delta_L} \\
{\cal A}_R = \tuple{\Sigma, Q_R, s_R, Q_R, \delta_R}
\end{eqnarray*}
Specifically, $Q_L\subseteq 2^Q$ contains only accessible states, $s_L=I$, and for all $L\in Q_L$ and $a\in \Sigma$:
\begin{equation*}
\delta_L(L,a) := \set{q'}{\exists q \in L, m \in M : \tuple{q,\tuple{a,m},q'} \in \Delta^*}.
\end{equation*}
Similarly, $Q_R\subseteq 2^Q$ contains only accessible states, $s_R=F$, and for all $R\in Q_R$ and $a\in \Sigma$:
\begin{equation*}
\delta_R(R,a) := \set{q'}{\exists q \in R, m \in M : \tuple{q',\tuple{a,m},q} \in \Delta^*}. 
\end{equation*}
\noindent{\bf Step 3-$\varepsilon$.\ } See Step~3.

\noindent{\bf Step 4-$\varepsilon$.\ } See Step~4. 

\noindent{\bf Step 5-$\varepsilon$.\ } Finally, we define ${\cal B} :=\tuple{{\cal M},{\cal A}_L,{\cal A}_R,\psi}$.

With these changes we can prove the following modification of Proposition~\ref{prop:well-defined}:
\begin{proposition}\label{prop:well-pre-defined2}
Let ${\cal A}^2$, ${\cal A}_L$, ${\cal A}_R$ and $\phi$ be defined as in Steps~1-$\varepsilon$,~2-$\varepsilon$~and~3-$\varepsilon$ above. Assume that $L,L'\in Q_L$, $R,R'\in Q_R$ and $a\in \Sigma$ satisfy
\begin{equation*}
L'=\delta_L(L,a) \text{ and } R=\delta_R(R',a)
\end{equation*}
and let $S=L\cap R$ and $S'=L'\cap R'$. If $S \not= \emptyset$, then:
\begin{enumerate}
\item  
for each state $p\in S$ there is a generalized transition $t=\tuple{p,\tuple{a,m},p'}\in \Delta^*$ with $p'\in S'$.
\item there is $c\in M$ such that for every generalized transition ${\tau=\tuple{q,\tuple{a,\lambda},q'}\in \Delta^*}$ with $q\in S$ and $q'\in S'$:
\begin{equation*}
\phi_S(q) c=\lambda \phi_{S'}(q').
\end{equation*}
\end{enumerate}
\end{proposition}
\begin{proof}
See the proof of Proposition~\ref{prop:well-defined}. \qed
\end{proof}

However, Proposition~\ref{prop:well-pre-defined2} is not enough to establish the correctness of Theorem~\ref{th:contribution}.
The problem is that in Step~4, and thus in Step~4-$\varepsilon$, we select an ordinary transition $t=\tuple{p,\tuple{a,m},p'}\in \Delta$ and not a generalized $\tuple{p,\tuple{a,m},p'}\in \Delta^*$ as in Proposition~\ref{prop:well-pre-defined2}.
The following lemma shows that we can select an ordinary transition in Step~4-$\varepsilon$. 

\begin{lemma}\label{lemma:general_to_ordinary}
Let $L$ be a state in ${\cal A}_L$, $R'$ be a state in ${\cal A}_R$, and $a\in \Sigma$.
If $R=\delta_R(R',a)$ and $L'=\delta_L(L,a)$ are such that, $S=R\cap L\neq \emptyset$
and $S'=L'\cap R'$, then there exists a transition $t=\tuple{p,\tuple{a,m},p'}\in \Delta$ such that $p \in S$ and $p'\in S'$.
\end{lemma}
\begin{proof}
Let $q\in S=R\cap L$. Then there is a state $q'\in R'$ with $\tuple{q,\tuple{a,m'},q'}\in \Delta^*$.
Since $q\in L$ it follows that $q'\in L'$. Hence $q'\in L'\cap R'=S'$. Fix a path $\pi$ form $q$ to $q'$
with label $\tuple{a,m'}$. Let $t=\tuple{p,\tuple{a,m},p'}$ be the unique transition on $\pi$ with input
label $a\in \Sigma$. It follows that $\tuple{q,\tuple{\varepsilon,m_1},p}\in \Delta^*$ and $\tuple{p',\tuple{\varepsilon,m_2},q'}\in \Delta^*$. Since $q\in L$, we get that $p\in L$. Since $q'\in R'$ we obtain that $p'\in R'$. Now the transition
$t$ with $p'\in R'$ shows that $p\in R=\delta_R'(R',a)$, and thus $p\in L\cap R=S$. Similarly, since $p\in L$
the transition $t$ witnesses that $p'\in L'=\delta_L'(L,a)$ and therefore $p'\in L'\cap R'$. \qed
\end{proof}

\begin{theorem}\label{th:contribution_general_constructive} 
$O_{{\cal B}'}=O_{\cal T}$.
\end{theorem}
\begin{proof}
If $u=\varepsilon$, then by convention we have that $O_{{\cal B}'}(\varepsilon)=e=O_{\cal T}(\varepsilon)$.
For $u\in \Sigma^+$, we can repeat the proof of Theorem~\ref{th:groups} where we use Proposition~\ref{prop:well-pre-defined2}
and Lemma~\ref{lemma:general_to_ordinary} instead of Proposition~\ref{prop:well-defined}. We omit the details.\qed
\end{proof}

\begin{proof}[of Theorem~\ref{th:contribution}]
By construction, ${\cal A}_L$ and ${\cal A}_R$ have at most $2^{|Q|}$ states. Now the result follows by Theorem~\ref{th:contribution_general_constructive}.\qed
\end{proof}

\section{Comparison to previous bimachine constructions}\label{SectionComparison}

In this section we show that our construction outperforms previous methods for building bimachines from functional transducers.
After a brief discussion of the standard bimachine construction in \cite{RS97} we introduce a class of real-time functional transducers with $n+2$ states for which the standard bimachine construction generates a bimachine with at least $\Theta(n!)$ states. 
The present construction based on the equalizer accumulation principle leads to 
$2^n + n +3$ states. 

\subsection{Classical Construction}\label{SubSecClCo}
A comprehensive outline of the standard bimachine construction can be found elsewhere,~\cite{RS97,Eil74,Sakarovitch09}.
Here we follow~\cite{CIAA2017} and provide only the basic ideas, stressing the main properties of the construction that are relevant for our discussion.

The classical construction of bimachines~\cite{Eil74} refers to the special case where ${\cal M}=\tuple{\Omega^*,\circ,\varepsilon}$ is the free monoid generated by an alphabet $\Omega$. As described in~\cite{RS97}, but see also the proofs in~\cite{Eil74,Berstel79,Sakarovitch09}, it departs from a pseudo-deterministic transducer, i.e. a transducer ${\cal T}=\tuple{\Sigma\times\Omega^*,Q,I,F,\Delta}$ that can be considered as a deterministic finite-state automaton over the new alphabet $\Sigma\times \Omega^*$. This means that $I$ contains a single state $i$ and $\Delta$ is a finite graph of a function $Q\times (\Sigma\times\Omega^*)\rightarrow Q$. 

The next step is the core of the construction. The goal is to construct an unambiguous transducer ${\cal T}'$ with transition relation $\Delta'$ equivalent to ${\cal T}$. This is achieved by specializing the standard determinization construction for finite-state automata: the sets generated by the determinization procedure are split into two parts, a single {\em guessed positive state} -- this is our positive hypothesis for the successful path to be followed, and a set of {\em negative states} -- these are the alternative hypotheses that must all fail in order for our positive hypothesis to be confirmed. Formally, the states in the resulting transducer are pairs $\tuple{p,N}\in Q\times 2^Q$. The initial state is $i'=\tuple{i,\emptyset}$. The algorithm inductively defines transitions in $\Delta'$ and states in $Q'$. Let $\prec_{lex}$ denote the lexicographic order on $\Omega^*$. 
Let $\tuple{p,N}$ be a generated state and $\tuple{p,\tuple{a,v},p'}\in \Delta$. 
Let 
$$N' := \{q'\,|\, \exists q\in N, \exists v':\tuple{q,\tuple{a,v'},q'}\in \Delta\} \cup \{q \,|\, \exists v'\prec_{lex} v:\tuple{p,\tuple{a,v'},q}\in \Delta\}.$$
If $p' \not\in N'$ we introduce a state $\tuple{p',N'}$ and a transition  
$$\tuple{\tuple{p,N},\tuple{a,v},\tuple{p',N'}}\in \Delta'.$$
Intuitively, this transition makes a guess about the lexicographically smallest continuation of the output that can be followed to a final state $f\in F$. Accordingly, all transitions that have the same input character, $a$, but lexicographically smaller output, are implicitly assumed to fail. 
To reflect this, we add those states to the set of negative hypotheses, $N'$. 
To maintain the previously accumulated 
 negative hypotheses along the path to $\tuple{p,N}$ the $a$-successors of $N$ are added to $N'$. Following these lines, the set of final states of ${\cal T}'$ is defined as:
 \begin{equation*}
 F'=\{\tuple{f,N}\,|\, f\in F \text{ and } N\cap F=\emptyset\}.
\end{equation*}
 
Note that $\tuple{f,N}$ becomes final only if $f\in F$ and there is no final state $n\in N$ 
 reached with smaller output on a parallel path.  
 It can be formally shown,~\cite{RS97}, that this construction indeed leads to an unambiguous transducer:
 \begin{equation*}
 {\cal T}'=\tuple{\Sigma\times \Omega^*,Q',\{i'\},F',\Delta'}
 \end{equation*}
equivalent to ${\cal T}$.

The final step is to convert the (trimmed part of) ${\cal T}'$ in an equivalent bimachine. This can be easily done by a determinization of ${\cal A}_L={\cal A}_{{\cal T}',D}$ and ${\cal A}_R={{\cal A}_{\cal{T}'}^{rev}}_{D}$ and defining an appropriate output function $\psi:Q_L\times \Sigma\times Q_R\rightarrow \Omega^*$. 

\subsection{A special class of transducers}

The following example introduces a class of transducers that demonstrates the advantages of the new bimachine 
construction. 

\begin{example}\label{ExampleTransducers}
We consider the class of transducers, see also Figure~\ref{fig:transducer}:
$${\cal T}_n := \tuple{ \Sigma^*_n\times\{1\}^*,Q,\{s\},\{f\},\Delta_n}$$
where $\Sigma_n  =  \{a_1,a_2,\dots,a_n\}$, $Q :=  \{s,q_1,q_2,\dots,q_n,f\}$ and 
the transition relation $\Delta_n :=  \Delta_{s,n}\cup \Delta_{Q_n} \cup \Delta_{f,n}$ has three types of transitions: 
$$\Delta_{s,n} :=  \{\tuple{s,\tuple{a_j,1^{i-1}},q_i} \,|\, i\le n, j \le n\}$$ 
are transitions from the start state 
$s$ to the states in the ``intermediate layer'' $Q_n:= \{q_1,q_2,\dots,q_n\}$. 
$\Delta_{Q_n}$ is the graph of the function $\delta_{Q_n}:Q_n\times \Sigma_n\rightarrow 1^* \times Q_n$ defined as 
$$\delta_{Q_n}(q_i,a_j) = \begin{cases}
\tuple{1^n,q_i} \text{ if } i\not \in \{1,j\}\\
\tuple{1^{n+j-1},q_j} \text{ if } i=1 \\
\tuple{1^{n-j+1},q_1} \text{ if } i=j(\neq 1).
\end{cases}$$
and describes transitions between states in the intermediate layer $Q_n$. 
Finally $\Delta_{f,n}$ is the graph of the function 
$$
\delta_{f,n}(q_i,a_j)=\begin{cases}
\tuple{1^{2n-i+1},f} \text{ if } j\le i\\
\text{undefined, else}
\end{cases}
$$
and describes transitions from the intermediate layer to the final state $f$. 
For a given input string each transducer ${\cal T}_n$ outputs a sequence of letters $1$. 
See Figure~\ref{fig:transducer} for an illustration. 
\end{example}
\begin{figure}
\includegraphics[width=0.9\textwidth]{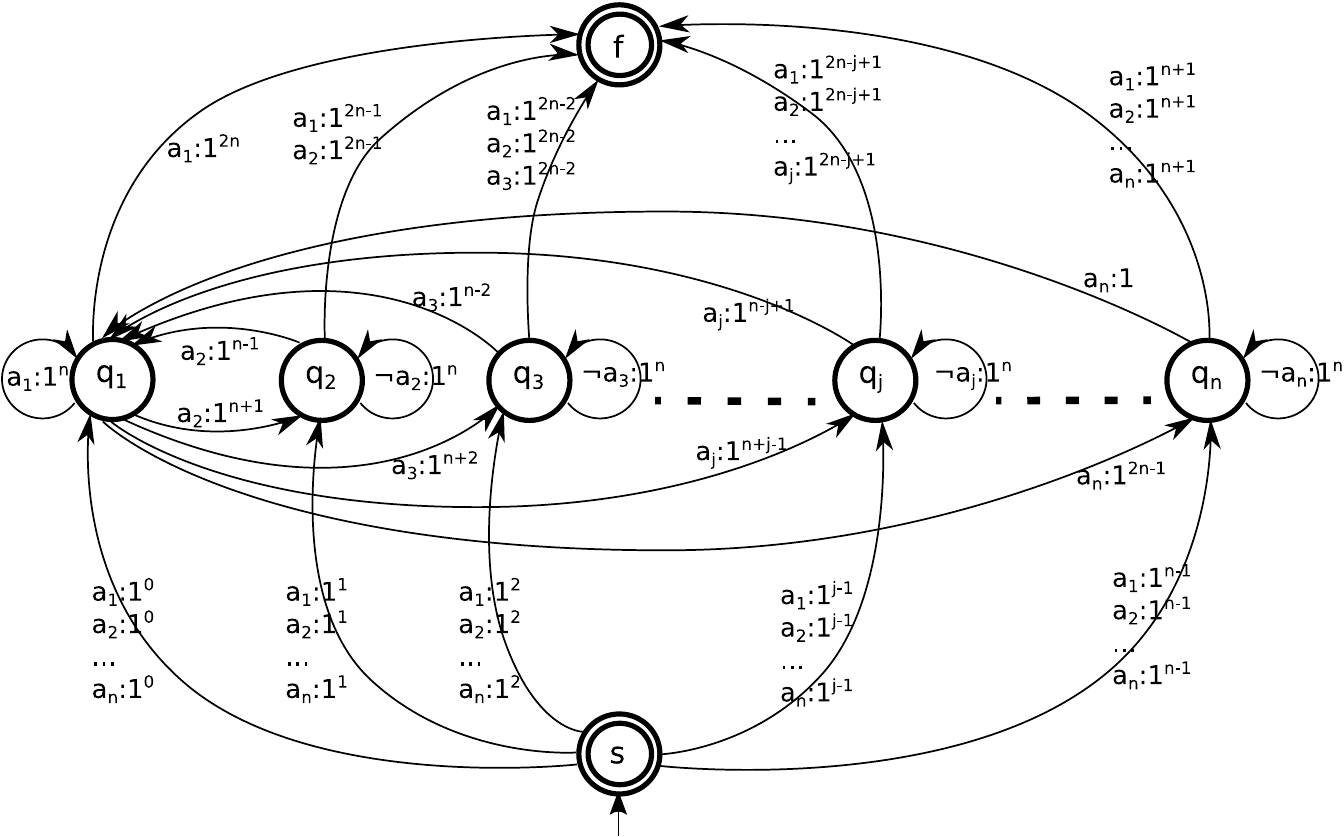}
\caption{The transducer ${\cal T}_n$. By $\neg a_j:1^n$ we denote the set of labels $a_i:1^n$ with $i \not= j$.}\label{fig:transducer}
\end{figure}
Before we analyze the set of bimachine states obtained from the standard and the new construction method we first summarize some 
simple properties.
\begin{lemma}\label{lemma:simple_properties}
${\cal T}_n$ is an ambiguous quasi-deterministic functional real-time transducer.
\end{lemma}
\begin{proof}
1. ${\cal T}_n$ is quasi-deterministic: Clearly, $\Delta_{s,n}$ is a graph of a function from $s\times\Sigma_n^*\times 1^*$. 
Furthermore the elements of $\Delta_{s,n}$ have first coordinate $s$, whereas all the entries in $\Delta_{f,n}\cup \Delta_{Q_n}$
have first coordinate in $Q_n$. Thus, it remains to show $\Delta_{f,n}\cup \Delta_{Q_n}$ is a graph of a function $\delta:Q_n \times \Sigma_n \times 1^* \rightarrow Q_n\cup \{s,f\}$. 
However, $\Delta_{f,n}$ and $\Delta_{Q_n}$ are graphs of functions $\delta_{f,n}$ and $\delta_{Q_n}$, respectively. Furthermore if $\tuple{q_i,a_j}$ is in the domain of both functions, then
\begin{equation*}
\delta_{f,n}(q_i,a_j)=\tuple{1^{2n-i+1},f} \text{ and } \delta_{Q_n}(q_i,a_j)\in \{\tuple{1^{n-i+1},q_1},\tuple{1^{n},q_i}\}.
\end{equation*}
Since $2n-i+1\ge n+1>n\ge n-i+1$ we conclude that the first coordinates of the results of these two functions are distinct.
Therefore, ${\cal T}_n$ is quasi-deterministic. 

2. For every path $\tuple{s,\tuple{\alpha,1^k},q_i}\in \Delta^*$ we have
\begin{equation*}
\alpha \in \Sigma_n^+ \text{ and } k=n (|\alpha|-1) +i-1.
\end{equation*}
This follows by a straightforward induction on the length of the path. We easily conclude that 
${\cal T}_n$ is a real-time quasi-deterministic functional transducer with language
\begin{equation*}
L({\cal T}_n)=\{\tuple{\alpha,1^{n|\alpha|}}\,|\, \alpha \in \Sigma_n^*\setminus \Sigma_n\}.
\end{equation*}
It is also clear that for each word in $\Sigma_n^+ \{a_1\}$ there are $n$ distinct successful paths in ${\cal T}_n$, thus ${\cal T}_n$ is ambiguous. 
\qed
\end{proof}
The next property is crucial for our analysis.
\begin{lemma}\label{permutation}
 For each permutation $\pi:Q_n\rightarrow Q_n$ there
 is a word $\alpha_{\pi}\in A^*$ such that for all $q,q'\in Q_n$:
 \begin{equation*}
 \tuple{q,\tuple{\alpha_{\pi},.},q'}\in \Delta^* \iff q'=\pi(q).
 \end{equation*}
\end{lemma}
\begin{proof}
For each $j$, the function $\delta_{Q_n}(q_i,a_j)$ induces a permutation on $Q_n$.
 Specifically, it maps each state $q\not \in \{q_1,q_j\}$ to itself, and it transposes $q_1$ and $q_j$.
 Thus, we can identify the action of $a_j$ on $Q_n$ with the transposition $\sigma_{1,j}:Q_n\rightarrow Q_n$:
\begin{equation*}
\sigma_{1,j}(q_i)=\begin{cases}
q_i \text{ if } i\not \in \{1,j\}\\
q_j \text{ if } i=1 \\
q_1 \text{ if } i=j.
\end{cases}
\end{equation*}
The transpositions $\sigma_{1,j}$ generate the entire permutation group on $Q_n$. Therefore for each
permutation $\pi:Q_n\rightarrow Q_n$, there is a word $\alpha_{\pi}\in \Sigma_n^*$ such that for all $q,q'\in Q_n$:
\begin{equation*}
\tuple{q,\tuple{\alpha_{\pi},.},q'} \in \Delta^*\iff q'=\pi(q').
\end{equation*}\qed
\end{proof}
We first look at the number of states of the bimachine obtained from the new construction based on the equalizer accumulation principle.
\begin{lemma}\label{lemma:new_constr_simple}
For each $n$:
\begin{enumerate}
\item the determinization of the underlying automaton of ${\cal T}_n$ generates $3$ states.
\item the determinization of the reversed underlying automaton of ${\cal T}_n$ generates $2^n+n$ states.
\end{enumerate}
\end{lemma}
\begin{proof}
1. The first part is immediate - the states are given by the three sets $\{s\}$, $Q_n$, $Q_n\cup \{f\}$.

2. As for the second, part, let $N_i=\{q_1,\dots,q_{i-1}\}$ for $1\leq i\leq n$. We note that the reversed determinization constructs:
\begin{itemize}
\item one initial state, $\{f\}$,
\item $n$ non-final states $Q_n\setminus N_i$.
\item each state $Q_n\setminus N_i$ gives rise to the sets $P\cup \{s\}$ where $P\subseteq Q_n$ where $|P|=n-i$.
\end{itemize}
The first two claims are obvious. For the third claim, first notice that for each state $q_i$ in the intermediate layer and each letter $a_j$ we always come back to the start state $s$ using a reverse transition. For the states of the intermediate layer $Q_n$, each letter $a_j$ induces a permutation (s.a.). Hence, ignoring $s$, backward transitions always lead from a set with $n-i$ states of $Q_n$ to another set with $n-i$ states of $Q_n$. Lemma~\ref{permutation} shows that each subset of $n-i$ states of $Q_n$ is reached from $Q_n\setminus N_i$ using a series of backward transitions, which proves the third claim. 
%
%
We conclude that the number of states generated by the determinization of the reversed underlying automaton of ${\cal T}_n$ is
$$
1 + n + \sum_{i=1}^n \binom{n}{n-i}=1+n+(2^n-1)=2^n + n.
$$
\qed
\end{proof}

The next proposition shows that the number of states of the bimachine obtained when using the classical construction for 
converting the transducers ${\cal T}_n$ is much larger than for the new construction. 
%
\begin{proposition}\label{prop:hard1}
Let ${\cal T}'_n$ be the trimmed result of the specialized determinization of ${\cal T}_n$ described in Subsection~\ref{SubSecClCo}. Then
\begin{enumerate}
\item ${\cal T}'_n$ has at least $(2n+3)2^{n-2}$ states.
\item The determinization of the underlying automaton of ${\cal T}_n'$ generates at least $n!+2$ states.
\item The determinization of the reversed underlying automaton of ${\cal T}_n'$ generates at least $2^n+n$ states.
\end{enumerate}
\end{proposition}
The proof of Proposition~\ref{prop:hard1} is given in the Appendix. The reason for the large number of states obtained when using the specialized determinization of ${\cal T}_n$ is that it distinguishes words $\alpha$ that correspond to different permutations on $Q_n$. This circumstance forces the determinization of the underlying automaton of ${\cal T}'_n$ to generate states that correspond to all the permutations of $Q_n$. Meanwhile, the determinization of the reversed underlying automaton of ${\cal T}'_n$ behaves timidly and, combinatorically, it is equivalent to the reversed determinization of ${\cal T}_n$ that we analyzed in Lemma~\ref{lemma:new_constr_simple}.

\begin{remark}
Similar arguments show that the bimachine construction introduced in Section~4 of ~\cite{CIAA2017} shares the same lower bound of $\Theta(n!)$ for the number of states when used for the class of transducers ${\cal T}_n$ introduced in Example~\ref{ExampleTransducers}.
\end{remark}

%

Summing up results we obtain the following theorem:
\begin{theorem}
For any positive integer $n$, on input ${\cal T}_n$ the standard bimachine construction generates at least $(n!+2)+(2^n+n)$ states, whereas
the construction proposed in Section~\ref{SecBimachineConstruction} generates a bimachine with $3+(2^n+n)$ states. 
\end{theorem}\qed

\section{Conclusion}\label{SecConclusion}

In this paper we introduced a new method for converting functional transducers into bimachines. 
The method can be applied to all functional transducers where the output monoid belongs to 
the class of mge monoids introduced in Section~\ref{SecMGE}. From our point of view this class, which was shown to be closed under Cartesian products, deserves interest on its own. An effective functionality decision procedure is given for transducers with output in effective mge monoids.

Unlike other methods, the new bimachine construction does not try to imitate in a bimachine run a particular path of the transducer. Rather, the squared output transducer and the principle of equalizer accumulation are used to obtain a unified view on all 
parallel paths of the transducer for a particular input. The bimachine output for a partial input sequence corresponds to the maximal output across all the outputs obtained in the parallel transducer paths. 

The advantage of the new construction method is space economy. The number of states of the bimachine is bounded by 
$2\cdot 2^{n}$ where $n$ is the number of transducer's states. In \cite{CIAA2017} it has been shown that the size of bimachines obtained when using other methods based on transducer path reconstruction is in $O(n!)$. Further we showed that the standard construction can achieve a worst case complexity $\Theta(n!)$ on a class of transducers with $n+2$ states.
We are not aware of any known method with better space complexity. 

We additionally introduced a natural amendment of the construction that can be directly applied to arbitrary non-real-time functional transducers. This construction avoids the complexity increase caused by the blow-up of the number of transitions when 
first building a real-time transducer.

\bibliographystyle{splncs03}
\bibliography{bibliography-2}

\section*{Appendix}

We prove Proposition~\ref{prop:hard1}. 

\begin{proof}
1. First, since $\tuple{s,\emptyset}$ is the initial state of ${\cal T}'_n$ the initial transitions, $\tuple{a_j,1^i}$ for any $a_j \in \Sigma_n$
generate the pairs:
\begin{equation*}
P_i=\tuple{q_i,N_i}, \text{ where } N_i=\{q_1,\dots,q_{i-1}\}.
\end{equation*}
for each $i\le n$. Furthermore, $q_i$ has a transition with $a_i$ to $f$ whereas none of the states in $N_i$ has a transition with $a_i$ to $f$. Therefore $P_i$ is also co-accessible in ${\cal T}'_n$.

Let $p\in Q_n$ and $S\subseteq Q_n$ be arbitrary such that $p\not \in S$. Then $|S|=i-1$ for some $i\le n$.
Therefore there is a permutation
$\pi:Q_n\rightarrow Q_n$ such that $\pi(q_i)=p$ and\footnote{If $C\subseteq Q_n$, $\pi(C)$ is the set of $\pi(c)$ for $c\in C$.} $\pi(N_i)=S$. Following Lemma~\ref{permutation} let $\alpha_{\pi}$ be a word that realizes the permutation $\pi$. Hence, starting from $P_i$ and following the word $\alpha_{\pi}$ the construction will define a state $P_{i,\pi}$ such that:
\begin{equation*}
P_{i,\pi}=\tuple{\pi(q_i),\pi(N_i)\cup \{f\}}=\tuple{p,S\cup \{f\}} \text{ or } P_{i,\pi}=\tuple{\pi(q_i),\pi(N_i)}=\tuple{p,S}.
\end{equation*}
Let $\beta$ be the word that corresponds to $\pi^{-1}$. Thus starting from $P_{i,\pi}$ and following the word $\beta$ the determinization of ${\cal T}_n$ must reach the state $\tuple{q_i,N_i}$ or $\tuple{q_i,N_i\cup \{f\}}$. Since $f$ has no transitions,
we deduce that both of them are co-accessible implying that also $P_{i,\pi}$ is co-accessible. 

Furthermore, if $i>1$, then $\alpha_{\pi}a_1$ generates also the pair $\tuple{p,S\cup \{f\}}$ since $a_1$ corresponds to the identity on $Q_n$ and every state has a transition with $a_1$ to $f$. On the other hand, if $q_1\not \in S$ and $l$ is the maximal index of a state $q_l\in S$, then $\alpha_{\pi}a_la_l$ generates first $\tuple{p,(S\setminus\{q_l\})\cup \{q_1,f\}}$ and then $\tuple{p,S\setminus \{f\}}$ because there is a transition from $q_1$ with $a_l$ to $q_l$ but (by the maximality of $l$) none of the states in $S-q_l$ has a transition with $a_l$ to $f$.

By the above discussion for each state $p\in Q_n$ and every set $\emptyset\subset S\subseteq Q_n\setminus \{q\}$, the pair $\tuple{p,S\cup \{f\}}$ is a state in ${\cal T}_n'$. Clearly, there are $n(2^{n-1}-1)$ such pairs. 
Next, we count the pairs $\tuple{p,S}$ such that $p\in Q_n$ and $S\subseteq Q_n$. First, if $S\neq \emptyset$, then $\tuple{p,S\cup \{f\}}$ is a state in ${\cal T}_n'$. Now, assume that $p=q_j$ and $\emptyset\neq S\subseteq N_j\setminus\{q_1\}$. Let $l$ me the maximal index of a state in $S$. Then the word $a_la_l$ maps $S\cup \{f\}$ to $S$. For each $j$ there are $2^{j-2}-1$ sets $\emptyset\neq S\subseteq N_j\setminus\{q_1\}$. Summing over $j=2,3\dots,n$ we get
\begin{equation*}
\sum_{j=2}^n (2^{j-2}-1)=2^{n-1} - 1 - (n-1)=2^{n-1} - n.
\end{equation*}
Finally, each of the pairs $\tuple{q_j,\emptyset}$ is also in ${\cal T}_n'$ and we have the initial state $\tuple{s,\emptyset}$ and $\tuple{q_j,N_j}$ for $j>1$ add up to:
\begin{equation*}
n(2^{n-1}-1) + 2^{n-1} - n + 2n= (n+1)2^{n-1}
\end{equation*}
non-final states.

Finally, if $S_j\subseteq N_j$ is arbitrary there is an accessible state $\tuple{q_j,S_j'}$ with $S_j'\setminus \{f\}=S_j$. Now a transition with $a_j$ shows that $\tuple{f,S_j \cup \{q_1\}}$ is final. Therefore for every set $S\subseteq N_n$ that contains $q_1$ there is final state $\tuple{f,S}$. Therefore ${\cal T}_n'$ has at least $2^{n-2}$ final states. Summing up we see that ${\cal T}_n'$ has at least:
\begin{equation*}
(n+1)2^{n-1} + 2^{n-2}=(2n+3)2^{n-2}
\end{equation*}
states.

2. For the second part, we prove that to different permutations $\pi:Q_n\rightarrow Q_n$ we can assign different states in the left automaton ${\cal A}_L$. To achieve this, let $\pi:Q_n\rightarrow Q_n$ be some permutation. We consider the $\alpha_{\pi}$ as defined in Lemma~\ref{permutation}. Without loss of generality we can assume that $|\alpha_{\pi}|\ge 1$. First, it is easy to see that
\begin{equation*}
\delta_L(\tuple{s,\emptyset},a_1)= \{\tuple{q_i,N_i} \,|\, i\le n\}.
\end{equation*} 
Next, following the word $\alpha_{\pi}$ each state $\tuple{q_i,N_i}$ will reach some non-final states and some final states.
The set of non-final states reached by $\tuple{q_i,N_i}$ will be:
\begin{equation*}
\emptyset\neq P_{i,\pi}\subseteq \{\tuple{\pi(q_i),\pi(N_i)}, \tuple{\pi(q_i),\pi(N_i)\cup \{f\}}\}.
\end{equation*}
Hence, the states $\tuple{\pi(q_i),S_i}\in P_{i,\pi}$ have their first coordinate in $Q_n$, 
whereas their second coordinate, say $S_i$, has the property $|S_i\cap Q_n|=i-1$.
Meanwhile, the final states reached by $\tuple{q_i,N_i}$ have a first coordinate $f\not\in Q_n$.

This analysis shows that for every $\pi$ we have
\begin{equation*}
\delta_L^*(\tuple{s,\emptyset},a_1\alpha_{\pi})=\bigcup_{i=1}^nP_{i,\pi} \cup {\cal F}_{\pi}
\end{equation*}  
where ${\cal F}_{\pi}\neq \emptyset$ is some set of elements with first coordinate\footnote{The precise form ${\cal F}_{\pi}$ can be expressed in terms of the pairs $P_{i,\pi}$ and the last character of $\alpha_{\pi}$. Yet, it is not important for our analysis.} $f$. Now, we claim that if $\pi'\neq \pi''$ are two permutations, then
\begin{equation*}
L_{\pi'}=\bigcup_{i=1}^nP_{i,\pi'} \cup {\cal F}_{\pi'} \text{ and } L_{\pi''}=\bigcup_{i=1}^nP_{i,\pi''} \cup {\cal F}_{\pi''}
\end{equation*}
are distinct sets. Indeed, let $j\le n$ be such that $\pi'(q_j)\neq \pi''(q_j)$. By the above analysis there is a pair $\tuple{\pi'(q_j),S'_j}\in L_{\pi'}$ with $|S'_j \cap Q_n|=j-1$. Now, if $\tuple{\pi'(q_j),S_j'}\in L_{\pi''}$, then $\pi'(q_j)=\pi''(q_i)$ for some $i$ and again by the above analysis we conclude that $|S_j'\cap Q_n| =i-1$. Hence, $i=j$ implying that $\pi'(q_j)=\pi''(q_j)$. This is a contradiction.

Therefore, all the states $L_{\pi'}$ are distinct states. All of them are final, since $|a_1\alpha_{\pi}|\ge 2$, and the domain of
${\cal T}_n$ contains all these words. Additionally, ${\cal A}_L$ must contain two more states, i.e. the initial state $\{\tuple{s,\emptyset}\}$ and also:
$$\delta_L(\{\tuple{s,\emptyset}\},a_1)=\{\tuple{q_i,N_i}\,|\, i\le n\}.$$
Hence, in total, ${\cal A}_L$ contains at least $n!+2$ states. 

3. Let $s_R=F'$ be the set of final states of ${\cal T}'_n$. For $j\le n$ let $R_j$ and $R_j'$ be the sets:
\begin{equation*}
R_j = \delta_R(s_R,a_j) \text{ and } R_j'=\{q\,|\,\exists S( \tuple{q,S}\in R_j)\}.
\end{equation*}
Then, we deduce that $R_j$ consists of all the pairs $\tuple{q_i,S}$ such that $i\ge j$, $\tuple{q_i,S}$ is accessible in ${\cal T}'_n$, $S-f \subseteq N_j$, and $i\ge j$. Therefore $R_j'=Q_n\setminus N_j$ since $a_j\not \in \Dom{\cal T}_n$.

Let $\pi$ be an arbitrary permutation on $Q_n$, $\pi:Q_n\rightarrow Q_n$ and let $\alpha_{\pi}$ that realizes $\pi$. Without loss of generality we may assume that $|\alpha|\ge 1$. We consider the sets:
\begin{equation*}
R_{j,\pi}=\delta_R^*(R_j,\alpha_{\pi}) \text{ and }R'_{j,\pi}=\{q\,|\, \exists S(\tuple{q,S}\in R_j)\}. 
\end{equation*}
Now, since $\alpha^{rev} a_j \in \Dom({\cal T}_n)$ we must have that $\tuple{s,\emptyset}\in R_{j,\pi}$ and it should be clear that:
\begin{equation*}
 R'_{j,\pi}=\{s\} \cup \pi(R_j').
 \end{equation*}
Now since $\pi$ was arbitrary, we deduce for any subset $R'\subset Q_n$, $R'\neq \emptyset$, there is a state, some state
$R$ in the right automaton ${\cal A}_R$ such that:
\begin{equation*}
R'=\{s\} \cup \{q\in Q_n\,|\, \exists S (\tuple{q,S}\in R)\}.
\end{equation*}
Therefore there are at least $2^n-1$ states $R$ that all contain $\tuple{s,\emptyset}$. Adding the initial state, $s_R$ and the $n$ states $R_j$ for $j\le n$, we conclude that the right automaton ${\cal A}_R$ has at least:
\begin{equation*}
1+n+(2^{n} -1)=2^n +n
\end{equation*}
states.
 \qed
\end{proof}
\end{document}